\documentclass[10pt,a4paper]{article}

\usepackage{algorithm}
\usepackage{algorithmic}
\usepackage{amsmath,amsfonts,amssymb}
\usepackage{array}
\usepackage[caption=false,font=normalsize,labelfont=sf,textfont=sf]{subfig}
\usepackage{textcomp}
\usepackage{stfloats}
\usepackage{url}
\usepackage{verbatim}
\usepackage{graphicx}
\usepackage{cite}
\usepackage{dsfont,color}
\usepackage{tikz}
\usepackage{bbm}
\usepackage{enumitem}
\usepackage{hyperref}

\hypersetup{hidelinks}

\usepackage[margin=2.5cm]{geometry}

\usepackage{amsthm}

\theoremstyle{plain}
\newtheorem{theorem}{Theorem}[section]
\newtheorem{lemma}[theorem]{Lemma}

\theoremstyle{definition}
\newtheorem{definition}[theorem]{Definition}

\theoremstyle{remark}

\newenvironment{keywords}
  {\par\noindent\textbf{Keywords:}\ }
  {\par}

\usepackage{multirow}
\usepackage{diagbox}

 \title{A Unified Fractional Regularization Framework for Sparse Recovery}

\makeatletter
\let\originalfnsymbol\@fnsymbol
\renewcommand{\@fnsymbol}[1]{%
  \originalfnsymbol{\numexpr#1+1\relax}}
\makeatother

\author{
  Yinhao Zhao\thanks{School of Information Science and Technology,
  ShanghaiTech University, Shanghai, China.
  Email: \texttt{zhaoyh22025@shanghaitech.edu.cn}.}
  \and
  Yaohua Hu\thanks{School of Mathematical Sciences,
  Shenzhen University, Shenzhen, China.
  Email: \texttt{mayhhu@szu.edu.cn}.}
  \and
  Hao Wang\thanks{School of Information Science and Technology,
  ShanghaiTech University, Shanghai, China.
  Email: \texttt{haw309@gmail.com}.}
}

\date{}

\begin{document}

\maketitle

\begin{abstract}
We propose a unified fractional regularization framework for sparse signal recovery based on the $\ell_1/\ell_p^q$ model. 
This model generalizes several widely used sparsity-promoting regularizers and provides additional flexibility through the parameters $p$ and $q$. 
Our main theoretical contribution is the characterization of the equivalence between the first-order stationary points of the $\ell_1/\ell_p^q$ formulation and the subtractive $\ell_1-\alpha\ell_p$ model, thereby offering a unified perspective on these nonconvex regularizers. 
In addition, we establish a new sufficient recovery condition under the Restricted Isometry Property (RIP), which shows that the proposed framework can provide relaxed recovery guarantees and improved robustness. 
To solve the resulting nonconvex problem, we develop a majorization--minimization (MM) algorithm and prove its convergence by using the Kurdyka--{\L}ojasiewicz (KL) property. 
Numerical experiments on sparse recovery problems with different sensing matrices and MRI reconstruction demonstrate that the proposed approach outperforms existing methods in recovery accuracy.
\end{abstract}


\section{Introduction}\label{sec1}

Sparse recovery aims to find a sparse solution $x\in\mathbb{R}^n$ (i.e., with most entries being zero) to the underdetermined linear system $Ax=b$, where 
$A\in\mathbb{R}^{m\times n}$  is the sensing matrix and $b\in\mathbb{R}^m$ is the observed signal. This problem arises in numerous applications including compressed sensing (CS), magnetic resonance imaging (MRI), and radar signal processing. The most direct formulation is the $\ell_0$-norm minimization:
\begin{equation}\label{prob:l0} 
\begin{aligned}
\operatorname{minimize} \quad & \|x\|_0 \quad
\operatorname{subject \ to} \quad & Ax = b,
\end{aligned}
\end{equation}
where 
$\|x\|_0 = \sum_{i=1}^n \mathbbm{1}_{\{x_i\neq 0\}}$ 
  counts the number of nonzeros. However, solving \eqref{prob:l0} is NP-hard \cite{natarajan1995sparse}, rendering it computationally intractable for large-scale problems. Consequently, a standard approach is to replace the non-convex $\ell_0$ norm with the convex $\ell_1$ norm, leading to the Basis Pursuit (BP) problem \cite{chen2001atomic}:
\begin{equation}\label{prob:l1}
\begin{aligned}
\operatorname{minimize} \quad & \|x\|_1 \quad
\operatorname{subject \ to} \quad & Ax = b.
\end{aligned}
\end{equation}

Extensive research has established that when the sensing matrix $A$ satisfies 
the Restricted Isometry Property (RIP) \cite{candes2005decoding} and the signal $x$ is sufficiently sparse, the $\ell_1$ minimizer can exactly recover $x$. 
This result has played a central role in the development of compressed sensing and has motivated the design of efficient optimization algorithms, such as FISTA \cite{beck2009fast} and ADMM \cite{boyd2011distributed}.

Importantly, a broad class of random matrices, including Gaussian, sub-Gaussian, and subsampled Fourier matrices, have been shown to satisfy the RIP with high probability \cite{baraniuk2007compressive, rudelson2008sparse}, which provides theoretical justification for their widespread use in practical applications.


In recent years, nonconvex surrogates such as $\ell_p (0<p<1)$ \cite{chartrand2007exact} and SCAD \cite{fan2001variable} have attracted attention for providing tighter approximations to the $\ell_0$ norm. Within this context, two prominent paradigms have emerged:

\begin{enumerate}
\item[(a)] 
Subtractive Models: The $\ell_1-\alpha\ell_p$ regularization ($p>1$) \cite{huo_boldsymboll_1-beta_2023}, defined as:  
\begin{equation}
\begin{aligned}
\text{minimize} \quad & \|x\|_1 - \alpha \|x\|_p \quad
\text{subject to} \quad & Ax = b,
\end{aligned}
\label{prob:l1-lp}
\end{equation}
where $\alpha>0$.
The parameter $\alpha$ controls the weight of the
subtracted $\ell_p$ term, and the objective approaches
$\|x\|_1$ as $\alpha\to0^+$.
Moreover, the $\ell_1-\alpha\ell_p$ model provides an effective approach for ratio-based sparsity regularization \cite{huo_boldsymboll_1-beta_2023}.

\item[(b)] Ratio Models: The $\ell_1/\ell_p$ regularization ($p>1$) \cite{RN547}, which is scale-invariant and highly effective for high-dynamic signals:
\begin{equation}
\begin{aligned}
\operatorname{minimize} \quad & \frac{\|x\|_1}{\|x\|_p} \quad
\operatorname{subject \ to} \quad & Ax = b.
\end{aligned}
\label{prob:l1/lp}
\end{equation}
This ratio shares key sparsity-inducing properties with the $\ell_0$ norm. Crucially, \eqref{prob:l1/lp} is scale-invariant: for any $\beta > 0$, $ \frac{\beta\|x\|_1}{\beta \|x\|_p} =  \frac{\|x\|_1}{\|x\|_p}$,  ensuring the optimization is insensitive to the magnitude of $x$. When $p=2$, the model simplifies to $\ell_1/\ell_2$, where the ratio $\ell_1/\ell_2$ was introduced as a sparsity measure in \cite{hoyer2004non}. Theoretical analysis shows that under a strong Null Space Property (sNSP), sparse solutions are local minimizers of \eqref{prob:l1/lp} \cite{RN471}. Furthermore, DC algorithms applied to $\ell_1/\ell_2$ exhibit linear convergence \cite{RN540}, and efficient solvers based on ADMM, bisection, and adaptive methods have been developed \cite{RN471, RN493, RN475, RN465, RN540, tao2023study}.

\end{enumerate}

Previous studies have investigated the strong null space property 
(sNSP) \cite{RN471} and DC algorithms
\cite{RN540, RN493, RN475, RN465} for ratio models.
More recent developments include the squared $\ell_1/\ell_2$
ratio model for sparse recovery \cite{jia2025sparse} and
extensions to block-sparse recovery \cite{zhan2025p}.

The connection between ratio objectives and parametric difference
problems is classical in fractional programming; see
Dinkelbach \cite{dinkelbach1967nonlinear} and the survey by
Frenk and Schaible \cite{frenk2005fractional}.
Suppose that $g(x)>0$ for all $x\in C$ and that the minimum
$r^*:=\min_{x\in C} f(x)/g(x)$ is attained. Then
\[
\min_{x\in C} \bigl\{f(x)-r^*g(x)\bigr\}=0,
\]
and minimizing $f(x)/g(x)$ and minimizing $f(x)-r^*g(x)$
over $C$ yield the same global minimizers.
Taking $f(x)=\|x\|_1$ and $g(x)=\|x\|_p$ yields the
corresponding relationship between the $\ell_1/\ell_p$ and
$\ell_1-\alpha\ell_p$ models when $\alpha=r^*$.
Here, we assume $b\neq 0$, so that the denominator is positive
on the feasible set $\{x:Ax=b\}$.
Dinkelbach's method updates the parameter iteratively, without
requiring prior knowledge of the optimal ratio value.


Motivated by these connections, we investigate under what
conditions, for any fixed $\alpha\in(0,n^{1-1/p})$,
a critical point of the $\ell_1-\alpha\ell_p$ model
corresponds to a critical point of a fractional sparsity model.
Figure~\ref{fig:ratio_difference} illustrates this relationship.
We further investigate RIP-based recovery conditions and
develop a tailored algorithm with convergence guarantees.

\begin{figure}[h!]
  \centering
  \begin{tikzpicture}[node distance=4cm, auto]
    \tikzstyle{mynode}=[
      draw,
      rounded corners,
      fill=blue!10,
      text width=3cm,
      align=center,
      minimum width=3cm,
      minimum height=1.5cm,
      inner sep=1pt
    ]

    \node[mynode] (A)
      {$\frac{\|x\|_1}{\|x\|_2}$};
    \node[mynode, right of=A, node distance=8cm] (B)
      {$\|x\|_1-\alpha^*\|x\|_2$};
    \node[mynode, below of=A, node distance=2cm] (C)
      {?};
    \node[mynode, below of=B, node distance=2cm] (D)
      {$\|x\|_1-\alpha\|x\|_p$};

    \draw[->, thick, orange] (A) -- (B)
      node[midway, above, sloped]
      {$\alpha^*=\frac{\|x^*\|_1}{\|x^*\|_2}$};

    \draw[->, thick, orange] (D) -- (C)
      node[midway, above, sloped]
      {For any fixed $0<\alpha<n^{1-1/p}$};
  \end{tikzpicture}
  \caption{
    Relationship between the $\ell_1/\ell_2$ and
    $\ell_1-\alpha^*\ell_2$ models when
    $\alpha^*=\|x^*\|_1/\|x^*\|_2$.
    For any fixed $\alpha\in(0,n^{1-1/p})$, we investigate
    under what conditions a critical point of the
    $\ell_1-\alpha\ell_p$ model corresponds to a critical
    point of a fractional sparsity model.
  }
  \label{fig:ratio_difference}
\end{figure}

Specifically, we consider the following generalized fractional model:
\begin{equation}
\label{prob:original}
\begin{aligned}
\operatorname{minimize} \quad &  \frac{\|x\|_1}{\|x\|_p^q} \quad
\operatorname{subject \ to} \quad & Ax = b,
\end{aligned}
\end{equation}
where $p>1$ and $0<q\le 1$.
When $q=1$, \eqref{prob:original} reduces to
\eqref{prob:l1/lp}.
For general $q$, the classical parametric reformulation
involves the difference objective
$\|x\|_1-\lambda\|x\|_p^q$.
Our analysis characterizes the correspondence between
critical points of \eqref{prob:original} and those of
$\ell_1-\alpha\ell_p$ problems through a point-dependent
parameter $\alpha$.

 Our main contributions are:
\begin{enumerate}
\item  {Theoretical Equivalence.} 
We establish a stationary-point relationship between
\eqref{prob:original} and $\ell_1-\alpha\ell_p$ regularization.
For a fixed $\alpha\in(0,n^{1-1/p})$, we characterize
conditions under which a critical point of the difference
model is also a critical point of \eqref{prob:original}
for a suitable $q\in(0,1]$.
\item  {Recovery Guarantee.} We derive a new RIP-based condition  for \eqref{prob:original}. This condition is weaker than many existing nonconvex guarantees \cite{chartrand2007exact}, ensuring robustness even under moderate coherence.
\item  {Efficient MM Algorithm.} We develop a tailored Majorization Minimization (MM) algorithm with an adaptive scheme. We prove its global convergence and local convergence rate.
\item {Superior Performance.} Numerical experiments on sparse signal recovery, time-domain signal reconstruction, and MRI reconstruction demonstrate that the proposed $\ell_1/\ell_p^q$ model outperforms the classical $\ell_1$, $\ell_1 - \ell_2$, and $\ell_1/\ell_2$ models in terms of recovery accuracy.
\end{enumerate}

The remainder of this paper is organized as follows. Section \ref{sec:notation} introduces notation and preliminary concepts. Section \ref{sec:equivalence_KL} presents the theoretical equivalence between \eqref{prob:original} and $\ell_1-\alpha\ell_p$ regularization. Section \ref{sec:rip} provides the  sufficient recovery condition under RIP.  Section \ref{sec:algorithm} develops our MM algorithm for solving \eqref{prob:original} and section \ref{sec:convergence} analyzes its convergence property. Section \ref{sec:experiments} provides extensive numerical results comparing our method against benchmarks. Finally, Section \ref{sec:conclusion} concludes the paper.

\section{Notation and preliminaries}
\label{sec:notation}

 Let $\mathbb{N}$ denote the set of positive integers. 
 $\mathbb{R}^n_+$ denotes the set of vectors in 
 $\mathbb{R}^n$ with all components nonnegative, while 
$\mathbb{R}^n_{++}$ denotes the set with all components strictly positive.
 For $x, y \in \mathbb{R}^n$, $\langle x, y \rangle$ denotes the standard inner product, and $x \circ y$ denotes the Hadamard product. We write $|x| = [|x_1|, \ldots, |x_n|]^\top$. We denote by $\mathbf{1} \in \mathbb{R}^n$ the vector with all components equal to $1$. The support of $x$ is $\text{supp}(x) = \{i : x_i \neq 0\}$, and $x$ is $k$-sparse if $|\text{supp}(x)| \le k$. The $\ell_p$-norm ($p > 0$) is $\|x\|_p = (\sum_{i=1}^n |x_i|^p)^{1/p}$, with $\|x\|_0 = |\operatorname{supp}(x)|$. For a nonempty closed set $\Omega \subseteq \mathbb{R}^n$, $\operatorname{dist}(x, \Omega) = \inf_{z \in \Omega} \|x - z\|_2$, and $\delta_\Omega(x)$ is the indicator function which is $0$ if $x \in \Omega$ and $+\infty$ otherwise. We use $\mathcal{U}(a,b)$ to denote the uniform distribution on $(a,b)$, and $\mathcal{N}(a,b)$ to denote the normal distribution with mean $a$ and variance $b$.

 \subsection{Nonsmooth Analysis}
 
 To handle the nonconvex and nonsmooth objective in \eqref{prob:original}, we employ the Clarke subdifferential. For a locally Lipschitz continuous function $f: \mathbb{R}^n \to \mathbb{R}$, the Clarke generalized directional derivative of $f$ at $x$ in the direction $v$ is defined as
 $$ f^\circ(x; v) = \limsup_{y \to x, t \downarrow 0} \frac{f(y + tv) - f(y)}{t}. $$
 The Clarke subdifferential $\partial f(x)$ is then given by $\partial f(x) = \{ \xi \in \mathbb{R}^n : \langle \xi, v \rangle \le f^\circ(x; v), \forall v \in \mathbb{R}^n \}$. A point $\bar{x}$ is a stationary point of $f$ if $0 \in \partial f(\bar{x})$. Note that for the convex (or concave) components in our model, the Clarke subdifferential coincides with the classical subdifferential (or the negative of the subdifferential of $-f$).

For a closed convex set $\Omega \subseteq \mathbb{R}^n$ and a point $\bar{x} \in \Omega$, the normal cone to $\Omega$ at $\bar{x}$ is defined as
\[
N_{\Omega}(\bar{x}) = \{ \xi \in \mathbb{R}^n : \langle \xi, x - \bar{x} \rangle \le 0, \; \forall x \in \Omega \}.
\]
For the constraint set $\Omega = \{x : Ax = b\}$, the normal cone is simply the orthogonal complement of the null space of $A$, i.e., $N_{\Omega}(x) = \operatorname{range}(A^\top)$ for any $x \in \Omega$.

 \subsection{KL    Properties}
 
 The Kurdyka--\L ojasiewicz (KL) property is essential for proving the convergence of descent methods for minimizing $F$. The convergence of the sequence $\{x^k\}$ is established by verifying the standard conditions for descent methods (see \cite{Bolte2010, Bolte2014, attouch2013convergence, attouch2009convergence}). If the descent algorithm ensures a sufficient descent of the objective function   and the subgradient is bounded by the successive iterates, satisfying conditions H1--H3 in \cite[Section 2.3]{attouch2013convergence}. Given the boundedness of $\{x^k\}$ and the fact that the objective $F$ is  a KL function, one can invoke the KL framework to conclude that the entire sequence $\{x^k\}$ converges to a stationary point of $F$.

 \begin{definition}[KL Property \cite{attouch2013convergence}] 
 A proper lower semicontinuous function $f$ satisfies the KL property at $\bar{x} \in \operatorname{dom} \partial f$ if there exist $\eta \in (0, +\infty]$, a neighborhood $U$ of $\bar{x}$, and a concave function $\phi(s) = cs^{1-\theta}$ for some $c>0$ and $\theta \in [0, 1)$, such that for all $x \in U \cap \{x : f(\bar{x}) < f(x) < f(\bar{x}) + \eta\}$, one has$$ \phi'(f(x) - f(\bar{x})) \operatorname{dist}(0, \partial f(x)) \ge 1. $$The value $\theta$ is referred to as the KL exponent.
 \end{definition}

%
%
%

The KL exponent $\theta$ plays a pivotal role in determining the convergence rate of the proposed algorithm. Specifically:
\begin{itemize}
\item If $\theta \in (0, 1/2]$, the sequence $\{x^k\}$ converges to $\bar{x}$ at a linear rate, i.e., $\|x^k - \bar{x}\| \le O(\rho^k)$ for some $\rho \in (0,1)$.
\item If $\theta \in (1/2, 1)$, the convergence is sublinear, satisfying $\|x^k - \bar{x}\| \le O(k^{-\frac{1-\theta}{2\theta-1}})$.
\end{itemize}

\section{The Proposed $\ell_1/\ell_p^q$ Framework}
\label{sec:equivalence_KL}

In this section, we present the mathematical formulation of our proposed models and establish their theoretical consistency. We focus on the $\ell_1/\ell_p$ ratio and its logarithmic variants, establishing their relationship in terms of stationary points and convergence properties.

\subsection{Problem Formulations}

Let $\Omega \subseteq \mathbb{R}^n$ be a closed convex set representing the feasible region. In the context of sparse signal recovery, typical choices include the noise-free constraint $\Omega = \{x : Ax = b\}$ or the noise-aware constraint $\Omega = \{x : \|Ax-b\|_2 \le \varepsilon\}$.

To promote sparsity in the solution, we first propose the following fractional constrained model:
\begin{equation}
  \begin{aligned}
    \underset{x\in \mathbb{R}^n}{\operatorname{minimize}} \quad & F(x):= \frac{\|x\|_1}{\|x\|^q_p}  
    \qquad \operatorname{subject \ to  }   &  x \in \Omega,
  \end{aligned}
  \label{prob.proposed} 
\end{equation}
where $p>1$ and $0 < q \leq 1$. To facilitate algorithmic design, we also consider a general transformation using a monotonic $C^1$ function $\phi: \mathbb{R}_{++} \to \mathbb{R}$ with $\phi'(t) > 0$:
\begin{equation}
  \begin{aligned}
    \underset{x\in \mathbb{R}^n}{\operatorname{minimize}} \quad & F_\phi(x):= \phi\left(\frac{\|x\|_1}{\|x\|^q_p}\right) 
        \qquad \operatorname{subject \ to  }   &   x \in \Omega. 
  \end{aligned}
  \label{prob.reform}
\end{equation}
Since $\phi$ is strictly increasing, the transformed problem \eqref{prob.reform} shares the same set of global (and local) minimizers as the original fractional model \eqref{prob.proposed}.
 A primary instance of \eqref{prob.reform} is the log-ratio model ($\phi = \ln$):
\begin{equation}
  \begin{aligned}
    \underset{x\in \mathbb{R}^n}{\operatorname{minimize}} \quad &  F_{\ln}(x):= \ln\|x\|_1 - q\ln\|x\|_p 
    \qquad \operatorname{subject \ to  }   & x \in \Omega. 
  \end{aligned}
  \label{prob.log}
\end{equation}

\subsection{Equivalence of Stationary Points}

We first relate the first-order stationarity conditions of the proposed fractional (ratio) model~\eqref{prob.proposed} and the well-known difference-based $\ell_1 - \alpha \ell_p$ model:
\begin{equation}
\begin{aligned}
  \underset{x\in \mathbb{R}^n}{\operatorname{minimize}}  \quad & G(x) := \|x\|_1 - \alpha \|x\|_p 
    \qquad \operatorname{subject \ to  }   & x \in \Omega.
\end{aligned}
\label{prob.l1lp}
\end{equation}

We now establish the first-order stationarity equivalence between the general transformed fractional model~\eqref{prob.reform} and the difference-based model \eqref{prob.l1lp}. This result provides a unified connection across various functional forms. 

\begin{theorem}[Stationarity Equivalence]\label{thm:equivalence_phi}
Let $\Omega\subset\mathbb{R}^n$ be closed convex,  $x^* \in \Omega \setminus \{0\}$ and let $\phi:\mathbb{R}_{++} \to \mathbb{R}$ be a $C^1$ function with $\phi'(t) > 0$ on $\mathbb{R}_{++}$. Then $x^*$ is a stationary point of  \eqref{prob.reform} if and only if it is a stationary point of  \eqref{prob.l1lp} with 
\begin{equation}\label{eq.settingq}
\alpha = q \frac{\|x^*\|_1}{\|x^*\|_p}.
\end{equation}
\end{theorem}
\begin{proof}
Let $N_{\Omega}(x)$ denote the normal cone to the feasible set $\Omega$ at $x$. Since $p > 1$, the $\ell_p$-norm is continuously differentiable on $\mathbb{R}^n \setminus \{0\}$.
  The first-order optimality condition for problem~\eqref{prob.l1lp} at $x^*$ reads
\begin{equation}\label{eq:stat_diff}
0 \in \partial \|x^*\|_1 - \alpha \nabla \|x^*\|_p + N_{\Omega}(x^*).
\end{equation}
 For the composite objective $F_\phi(x) = \phi\bigl(F(x)\bigr)$ with $F(x) = \|x\|_1 / \|x\|_p^q$, the stationarity condition of problem~\eqref{prob.reform} at $x^*$ is
\begin{equation}\label{eq:stat_ratio_phi}
0 \in \phi'\bigl(F(x^*)\bigr) \cdot \partial F(x^*) + N_{\Omega}(x^*).
\end{equation}
Using the quotient rule for the subdifferential of $F(x)$ at $x^* \neq 0$ yields
\begin{equation}\label{eq:stat_ratio_expanded}
\partial F(x^*) = \frac{1}{\|x^*\|_p^q} \partial \|x^*\|_1 - q \frac{\|x^*\|_1}{\|x^*\|_p^{q+1}} \nabla \|x^*\|_p.
\end{equation}
Substituting into~\eqref{eq:stat_ratio_phi} gives
\begin{equation}\label{eq:stat_ratio_full}
0 \in \phi'\bigl(F(x^*)\bigr) \left[ \frac{1}{\|x^*\|_p^q} \partial \|x^*\|_1 - q \frac{\|x^*\|_1}{\|x^*\|_p^{q+1}} \nabla \|x^*\|_p \right] + N_{\Omega}(x^*).
\end{equation}

 Define the positive scalar
\begin{equation}
\kappa := \phi'\bigl(F(x^*)\bigr) / \|x^*\|_p^q > 0.
\end{equation}
Then~\eqref{eq:stat_ratio_full} can be rewritten as
\begin{equation}\label{eq:stat_final_step}
0 \in \partial \|x^*\|_1 - \left( q \frac{\|x^*\|_1}{\|x^*\|_p} \right) \nabla \|x^*\|_p + \frac{1}{\kappa} N_{\Omega}(x^*).
\end{equation}
Since  $\frac{1}{\kappa} N_{\Omega}(x^*) = N_{\Omega}(x^*)$  for any $\kappa > 0$, 
we have 
$
\alpha = q \frac{\|x^*\|_1}{\|x^*\|_p}
$
transforms~\eqref{eq:stat_final_step} exactly into~\eqref{eq:stat_diff}.

Conversely, starting from~\eqref{eq:stat_diff} with this choice of $\alpha$ and reversing the steps (dividing by the positive factor $\phi'(F(x^*)) / \|x^*\|_p^q$) yields~\eqref{eq:stat_ratio_phi}. This establishes the equivalence.
\end{proof}

\section{ RIP Analysis}
\label{sec:rip}

In this section, we investigate fundamental properties of the proposed $\ell_1/\ell_p^q$ model. We establish recovery guarantees under the restricted isometry property (RIP).

The RIP is a cornerstone framework in compressed sensing. A measurement matrix $ {A} \in \mathbb{R}^{m \times n}$ is said to satisfy the RIP of order $k$ if there exists a constant $\delta_k \in (0,1)$, known as the restricted isometry constant (RIC), such that
\begin{equation}\label{rip}
    (1 - \delta_k) \|x\|_2^2 \leq \| {A}x\|_2^2 \leq (1 + \delta_k) \|x\|_2^2
\end{equation}
holds for any $k$-sparse vector $x \in \mathbb{R}^n$ (i.e., $\|x\|_0 \leq k$). The constant $\delta_k$ is the smallest value satisfying \eqref{rip} for a given order $k$.

To derive our recovery guarantee, we first introduce some notation. For a vector $x \in \mathbb{R}^n$ and an index set $T \subset \{1, \dots, n\}$, $x_T$ denotes the vector obtained by keeping the entries of $x$ indexed by $T$ and setting others to zero. The complement is $T$. Furthermore, $x_{(k)}$ denotes the vector obtained by keeping only the $k$ largest entries (in absolute value) of $x$ and setting the rest to zero, and $x_{(-k)} = x - x_{(k)}$.

We then present the main recovery guarantee.

\begin{theorem}\label{thm:rip}
Let $b = Ax^* + e$, where $x^*$ is a $k$-sparse signal and $\|e\|_2 \leq \varepsilon$ for some $\varepsilon > 0$. For $p > 1$, assume $\gamma \geq \frac{\|x^*\|_1}{\|x^*\|_p}$, and the measurement matrix $ {A}$ satisfies the RIP \eqref{rip} with a restricted isometry constant $\delta_{2k}$ that obeys
\begin{equation}
  \delta_{2k} < \hat\delta := \frac{1}{\sqrt{1 + k^{\frac{2}{\min\{p, 2\}}-2} \left( q\gamma t_0 + c_0  \right)^2}},
  \label{eq:RIPdelta}
\end{equation}
where $c_0 := q\gamma + k^{1-\frac{1}{p}}$ and  $t_0$ is the unique positive root from Lemma~\ref{lemma:2}. Then, any solution $\bar{x}$ of the noise-aware formulation for \eqref{prob.proposed} satisfies the error bound
\begin{equation}
  \|\bar{x} - x^*\|_2 \leq C \varepsilon,
  \label{eq:error_bound}
\end{equation}
where the constant $C$ is given by
\begin{equation}
  \begin{aligned}
    C &= \frac{2 \sqrt{1 + \delta_{2k}}}{1 -  \frac{\delta_{2k}}{\hat\delta } }   \times \left(1 + k^{\frac{1}{\min\{p, 2\}}} k^{-\frac{1}{2p} - \frac{1}{2}}  c_0^{\frac{1}{2}} \left( 1 + \frac{q\gamma t_0}{2   c_0  } \right) \right).
  \end{aligned}
  \label{eq:RIPC}
\end{equation}
 
\end{theorem}

This result is adapted from \cite[Theorem~1]{RN547}, with a modified auxiliary lemma.

\begin{lemma}\label{lemma:1}
Let $x \in \mathbb{R}^n$ be a $k$-sparse signal such that $\|Ax - b\|_2 \leq \varepsilon$ and $\frac{\|x\|_1}{\|x\|_p} \leq \gamma$ for some $\gamma > 0$. If $\bar{x}$ is a solution of the noise-aware formulation for \eqref{prob.proposed} and we define $h = \bar{x} - x$, then for $0 < q \leq 1$,
\begin{equation}
\|h_{(-k)}\|_1 \leq (q\gamma + k^{1-\frac{1}{p}})\|h_{(k)}\|_p + q\gamma \|h_{(-k)}\|_p.
\label{eq:lemma1_ineq}
\end{equation}
\end{lemma}

\begin{lemma}\label{lemma:2}
For any $p > 1$, $k \geq 1$, $0 < q \leq 1$, and $\gamma > 0$, consider the function $f(t) = t^p - q\gamma k^{-\frac{p-1}{p}} t - q\gamma k^{-\frac{p-1}{p}} - 1$. Then the equation $f(t) = 0$ has a unique positive solution $t_0$ in $[0, +\infty)$. Moreover, let $t^* = k^{-\frac{1}{p}} \left(\frac{q\gamma}{p}\right)^{\frac{1}{p-1}}$, the function $f(t)$ is decreasing on $(0, t^*)$ and increasing on $(t^*, +\infty)$, and \(t_0 > (kp)^{-\frac{1}{p}} (q \gamma)^{\frac{1}{p-1}}.\)
\end{lemma}

\begin{proof}[Proof of Theorem~\ref{thm:rip}]
We follow the proof of \cite[Theorem~1]{RN547}, 
replacing \cite[Lemma~3]{RN547} and \cite[Lemma~4]{RN547} 
with Lemma~\ref{lemma:1} and Lemma~\ref{lemma:2}, respectively. 
The former introduces an additional factor $q$ via the inequality $(1+t)^q \le 1+qt$ for \( t \ge 0 \) and \( 0 < q \le 1 \), while the latter corresponds to a direct rescaling of $\gamma$. As a result, $\gamma$ is replaced by $q\gamma$ in the constant $C$, yielding the desired conclusion.
\end{proof}

In the following discussion, we adopt the canonical choice $\gamma = k^{1-\frac{1}{p}}$, which allows us to express the sufficient recovery condition explicitly in terms of $(p,q,k)$. Under this parameterization, the equation governing $t_0$ simplifies such that $t_0$ depends only on $(p,q)$, decoupling from the ambient dimension $k$. 
When $1 < p \leq 2$, the upper bound $\hat{\delta}$ simplifies to
\( \hat{\delta} = \frac{1}{\sqrt{1 + (q t_0 + q + 1)^2}}, \)
which is independent of $k$. For the special case $p = 2$ and $q = 1$, we have $t_0 = 2$, leading to $\hat{\delta} = \frac{1}{\sqrt{17}}$. Conversely, when $p > 2$, the expression for $\hat{\delta}$ becomes
\( \hat{\delta} = \frac{1}{\sqrt{1 + k^{1-\frac{2}{p}} (q t_0 + q + 1)^2}}. \)
In this scenario, since the exponent $1-\frac{2}{p}$ is positive, $\hat{\delta}$ decreases as $k$ increases.

\begin{figure}[htbp!]
\centering
\subfloat[$\hat{\delta}$ vs $k$]{
\includegraphics[width=0.47\linewidth]{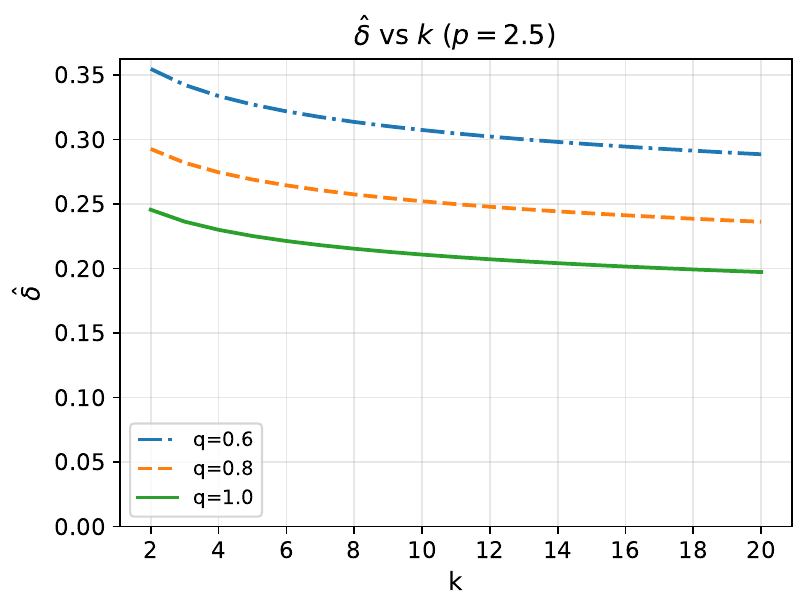}
\label{fig:delta_vs_k}
}
\hfill
\subfloat[$C$ vs $k$]{
\includegraphics[width=0.47\linewidth]{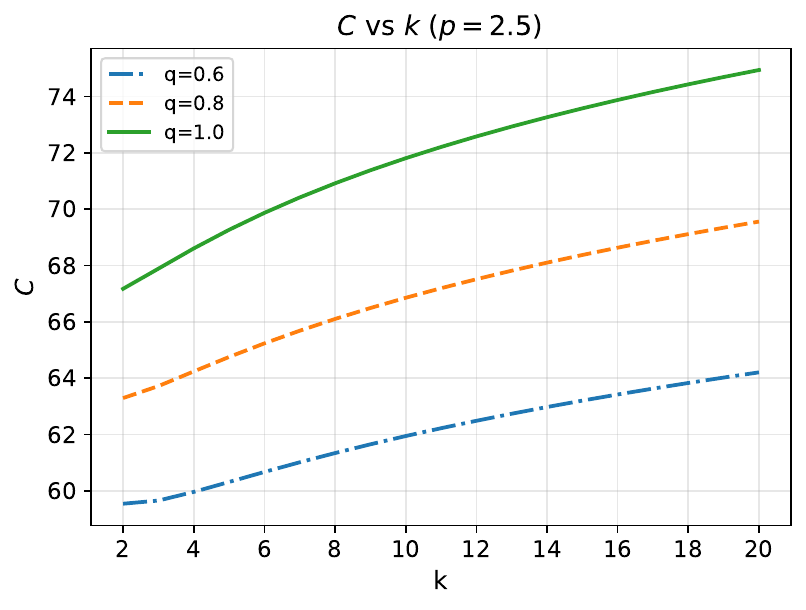}
\label{fig:C_vs_k}
}
\caption{Behavior of $\hat{\delta}$ (and implicitly $\delta_{2k}$) and the bound $C$ as functions of $q$ and $k$, with $p=2.5$ and $\gamma = k^{3/5}$ fixed.}
\label{fig:delta_and_C_trends}
\end{figure}

To isolate and clearly visualize the individual effects of the sparsity-related parameter $q$ and the dimension $k$, we fix $p=2.5$ and define $\delta_{2k} = 0.9 \hat{\delta}$.
Figure \ref{fig:delta_and_C_trends} illustrates the behavior of $\hat{\delta}$ (and implicitly $\delta_{2k}$) and the bound $C$ as functions of $q$ and $k$, with $p=2.5$ and $\gamma = k^{1-\frac{1}{p}}$ fixed. 
Figure \ref{fig:delta_vs_k} shows that $\hat{\delta}$ decreases as $k$ increases for $p > 2$. 
Since a larger $\hat{\delta}$ leads to a less restrictive RIP condition, this indicates that the recovery condition becomes more stringent as $k$ increases. Moreover, by introducing $q$, our model enlarges $\hat{\delta}$, thereby relaxing the RIP requirement.
Figure \ref{fig:C_vs_k} depicts the bound $C$ increasing with $k$, indicating that the error constant becomes larger as the sparsity $k$ increases. 
Since a smaller $C$ yields a tighter error bound, this implies degraded stability for larger $k$. Notably, the introduction of $q$ reduces $C$, leading to improved recovery accuracy.


Compared with the standard $\ell_1/\ell_p$ model \cite{RN547}, which corresponds to the case $q=1$, our generalized $\ell_1/\ell_p^q$ model with $q \in (0,1]$ introduces an additional degree of freedom. The presence of $q$ leads to a looser required bound on $\delta_{2k}$ and a potentially smaller error constant $C$ under the same $(k,p,\gamma)$ conditions. This theoretical advantage highlights the flexibility and potential improved performance of the proposed fractional model.

\section{Proximal MM algorithm for the Log Transformation} 
\label{sec:algorithm}

\subsection{Majorization of the Log Objective}

We now focus on solving \eqref{prob.log}. Although the log-transformed formulation removes the ratio structure,
the objective function remains nonconvex, since $\ln\|x\|_1$ is concave
and $-\ln\|x\|_p$ introduces additional nonlinearity.
To handle this difficulty in a principled manner, we adopt a
majorization--minimization (MM) framework. 

The key idea is to construct, at each iteration, a convex surrogate
function that globally upper-bounds the original objective and
coincides with it at the current iterate. Minimizing this surrogate
guarantees descent of the true objective.

The objective consists of two parts:
\[
F(x) = f(x) + g(x),
\quad
f(x) := \ln\|x\|_1,
\quad
g(x) := -q \ln\|x\|_p.
\]

Since $\ln(\cdot)$ is concave, $f(x)$ admits the global upper bound
\begin{equation}\label{eq:l1_major}
\ln\|x\|_1
\le
\ln\|x^{k}\|_1
+
\sum_{i=1}^n
\frac{|x_i|-|x_i^{k}|}{\|x^{k}\|_1}.
\end{equation}
Dropping constants yields a weighted $\ell_1$ term.

The gradient of \(g(x)\) is Lipschitz continuous on the considered region if this region is bounded away from the origin.
Therefore, there exists $L>0$ such that
\begin{equation}\label{eq:lipschitz}
g(x) \le g(x^{k}) + \langle \nabla g(x^{k}), x-x^{k} \rangle
+ \frac{L}{2}\|x-x^{k}\|_2^2.
\end{equation}

\subsection{Proximal MM Subproblem}

We formulate the considered problem as the following composite optimization problem
\begin{equation}\label{prob.general} 
\min_{x\in \Omega} 
F(x) := f(h(x)) + g(x),
\end{equation}
where the following assumptions hold:

\begin{itemize}
\item[(A0)] $\Omega$ is a closed convex set; 
\item[(A1)] $f:\mathbb{R}  \to (-\infty,+\infty]$ is proper, closed and concave;
\item[(A2)] $h: \mathbb{R}^n \to \mathbb{R}$ is proper, closed and convex; 
\item[(A3)] $g:\mathbb{R}^n \to \mathbb{R}$ is continuously differentiable;
\item[(A4)] $\nabla g$ is $L$-Lipschitz continuous, i.e.,
\[
\|\nabla g(x)-\nabla g(y)\|
\le L\|x-y\|, \quad \forall x,y.\]
\end{itemize}
The minimizer $x^*$ of \eqref{prob.general} satisfies the first-order optimality condition 
\begin{equation} \label{optimal.1} 
0 \in f'(h(x^*)) \partial h(x^*) + \nabla g(x^*) + N_\Omega(x^*).
\end{equation} 

Let $\beta > L$ be a fixed constant.  
Define the quadratic majorization function
\begin{equation}
\begin{aligned} 
& Q_\beta(x;x^{k})   =  F(x^{k}) + f'(h(x^{k}))( h(x) - h(x^{k}) )  + \langle \nabla g(x^{k}), x-x^{k}\rangle
+
\frac{\beta}{2}\|x-x^{k}\|^2.
\end{aligned} 
\end{equation}
The next iterate $x^{k+1}$ is then computed as the minimizer of $Q_\beta(x;x^{k})$ on $\Omega$: 
\[\min_x \  Q_\beta(x;x^{k}) \quad \text{s.t. } x\in \Omega.\]

 The proposed algorithm for solving \eqref{prob.general} is summarized in  Algorithm~\ref{alg.1}. 
 
\begin{algorithm}[H]
\caption{Proximal MM for Log-Ratio Minimization}\label{alg.1} 
\begin{algorithmic}[1]
\STATE Initialize $x^{(0)} \neq 0$ such that $ x^{(0)} \in \Omega$
\FOR{$k=0,1,2,\dots$}
\STATE Solve subproblem  for  $x^{k+1}$
$$x^{k+1} \gets \arg\min_{x\in\Omega} \  Q_\beta(x;x^{k}). $$
  \ENDFOR
\end{algorithmic}
\end{algorithm}

\subsection{Implementation} 

For solving the logarithm-transferred problem \eqref{prob.log}, the corresponding subproblem is:
\begin{equation}\label{eq:subproblem}
\begin{aligned}
\min_x& \ 
  Q_\beta(x;x^{k}) = 
\frac{1}{\|x^{k}\|_1}\sum_{i=1}^n  |x_i|
+ \frac{\beta}{2}\|x - x^{k}\|_2^2 + C + \langle \nabla g(x^{k}), x \rangle  \\
  \text{s.t. }  &  \ Ax=b, 
\end{aligned}
\end{equation}
where the constant $C = \ln\|x^{k}\|_1- 1 + g(x^{k}) - \langle \nabla g(x^{k}),  x^{k} \rangle.$

This is a convex optimization problem with a weighted $\ell_1$ objective
and a proximal quadratic term, which can be efficiently solved by
commercial solvers such as Gurobi \cite{gurobi}. 
Since 
\[ 
\nabla g(x^{k}) = -\frac{q}{\|x^{k}\|_p} \nabla \| x^{k}\|_p,
\]
by omitting the constant, the subproblem $Q$ can be equivalently scaled to  
\[\begin{aligned}
 \tilde Q(x;x^{k}) =  \ \|x\|_1
- q \frac{\|x^{k}\|_1}{\|x^{k}\|_p} \langle \nabla \| x^{k}\|_p, x\rangle  + \frac{\beta\|x^{k}\|_1}{2}\|x - x^{k}\|_2^2.
 \end{aligned}
\]

The ADMM for solving the $\ell_1/\ell_p$ problem in \cite{RN547} can be viewed as 
a special case of our algorithm with $q \equiv 1$ and 
a dynamically changing proximal parameter $\beta\|x^{k}\|_1$ instead of a fixed $\beta$. 
However,  the analysis in \cite{RN547} cannot be applied here to our proposed algorithm.

\section{Convergence Analysis}
\label{sec:convergence}
 
 We now analyze the convergence properties of the proposed proximal MM  Algorithm~\ref{alg.1} for solving \eqref{prob.general}.
We begin by establishing a monotonic descent property of 
objective, which follows directly from the majorization structure of the
surrogate function.

 We first  establish the fundamental descent property of the algorithm.

\begin{theorem}[Sufficient Decrease Property]\label{thm.decent} 
Let $\{x^{k}\}$ be generated by the proximal MM algorithm with $\beta > L$. Then for all $k$,
\begin{equation}
F(x^{k+1}) \le F(x^{k}) - \frac{\beta-L}{2} \|x^{k+1} - x^{k}\|^2.
\end{equation}
\end{theorem}

\begin{proof}
By the concavity of $f$, we have for any $x \in \Omega$:
\[
f(h(x)) \le f(h(x^{k})) + f'(h(x^{k}))(h(x) - h(x^{k})).
\]
On the other hand, since $\nabla g$ is $L$-Lipschitz continuous, the descent lemma gives:
\[
g(x) \le g(x^{k}) + \langle \nabla g(x^{k}), x - x^{k} \rangle + \frac{L}{2}\|x - x^{k}\|^2.
\]
Adding these two inequalities yields for any $x \in \Omega$:
\begin{equation}\label{eq.temp1} 
F(x) \le Q_\beta(x; x^{k}) - \frac{\beta-L}{2} \|x - x^{k}\|^2.  
\end{equation} 
By definition of $x^{k+1}$ as the minimizer of $Q_\beta(\cdot; x^{k})$ on $\Omega$, we have
\begin{equation}\label{eq.temp2} 
Q_\beta(x^{k+1}; x^{k}) \le Q_\beta(x^{k}; x^{k}) = F(x^{k}).  
\end{equation} 
Taking $x = x^{k+1}$ in \eqref{eq.temp1}  and combining with \eqref{eq.temp2} gives:
\[
F(x^{k+1})  
\le F(x^{k}) - \frac{\beta-L}{2} \|x^{k+1} - x^{k}\|^2,
\]
which completes the proof.
\end{proof}

\subsection{Global Convergence}

\begin{theorem}[Global Convergence]\label{thm.convergence} 
Assume that $F$ is bounded below. Then:

\begin{enumerate} 
\item [(i)] $\{F(x^{k})\}$ is monotonically decreasing and convergent;
\item [(ii)] $\sum_{k=0}^\infty \|x^{k+1}-x^{k}\|^2 < \infty$;
\item [(iii)] $\|x^{k+1}-x^{k}\| \to 0$;
\item [(iv)] Every cluster point of $\{x^{k}\}$ is a stationary point of \eqref{prob.general}, i.e., satisfies the optimality condition \eqref{optimal.1}.
\end{enumerate}
\end{theorem}

\begin{proof}
From the sufficient decrease inequality with $c = \frac{\beta-L}{2} > 0$, we have
\[
F(x^{k})-F(x^{k+1}) \ge c\|x^{k+1}-x^{k}\|^2.
\]
Summing from $k=0$ to $N$ gives
\[
F(x^{(0)})-F(x^{(N+1)}) \ge c\sum_{k=0}^N \|x^{k+1}-x^{k}\|^2.
\]
Since $F$ is bounded below, letting $N\to\infty$ yields
\[
\sum_{k=0}^\infty \|x^{k+1}-x^{k}\|^2 < \infty,
\]
which implies $\|x^{k+1}-x^{k}\|\to 0$. This proves (1)-(3).

To prove (4), let $x^*$ be any cluster point of $\{x^{k}\}$, and consider a subsequence $\{x^{(k_j)}\}$ such that $x^{(k_j)} \to x^*$. From the optimality condition of the subproblem at iteration $k_j$, there exists $s_{k_j+1} \in \partial h(x^{(k_j+1)})$ such that
\begin{equation}\label{eq.temp.limit} 
\begin{aligned} 
\beta (x^{(k_j+1)}-x^{(k_j)}) - \nabla g(x^{(k_j)}) - f'(h(x^{(k_j)})) s_{k_j+1}  \; \in N_\Omega(x^{(k_j+1)}).
\end{aligned}
\end{equation} 

Since $\|x^{(k_j+1)}-x^{(k_j)}\| \to 0$, we have $x^{(k_j+1)} \to x^*$ as well. The subdifferential $\partial h$ of a closed convex function is locally bounded and has closed graph. Therefore, the sequence $\{s_{k_j+1}\}$ is bounded and, by passing to a further subsequence if necessary, we may assume that $s_{k_j+1} \to s_*$ for some $s_* \in \mathbb{R}^n$. By the closedness of the graph of $\partial h$, we have $s_* \in \partial h(x^*)$.

Taking the limit in \eqref{eq.temp.limit} as $j \to \infty$, and using the continuity of $\nabla g$, the continuity of $f'$ at $h(x^*)$ (since $f$ is concave, $f'$ is continuous wherever it exists), and the outer semicontinuity of the normal cone $N_\Omega$ (which follows from the closedness of $\Omega$), we obtain
\[
-\nabla g(x^*) - f'(h(x^*)) s_* \in N_\Omega(x^*).
\]
Since $s_* \in \partial h(x^*)$, this is equivalent to
\[
0 \in f'(h(x^*)) \partial h(x^*) + \nabla g(x^*) + N_\Omega(x^*),
\]
which is exactly the first-order optimality condition \eqref{optimal.1}. Hence $x^*$ is a stationary point of \eqref{prob.general}.
\end{proof}

\subsection{Convergence Rates under KL}

To apply the KL framework, we need to establish a relative error condition that links the distance between consecutive iterates to the size of a subgradient at the new iterate.

\begin{lemma}[Relative Error Estimate]
Assume that  $h$ is Lipschitz continuous with constant $L_h$ on $\Omega$ and $f'$ is Lipschitz continuous with constant $L_f$ for $h$ on $\Omega$. Then there exists a constant $C > 0$ such that for all $k$,
\begin{equation}
\mathrm{dist}(0, \partial F(x^{k+1})) \le C \|x^{k+1} - x^{k}\|.
\end{equation}
\label{lemma:rel_err_est}
\end{lemma}

\begin{proof}
From the optimality condition of the subproblem, there exists $s_{k+1} \in \partial h(x^{k+1})$ such that
\begin{equation}\label{eq:relerr_optimality}
\beta (x^{k+1} - x^{k}) - \nabla g(x^{k}) - f'(h(x^{k})) s_{k+1} \in N_\Omega(x^{k+1}).
\end{equation}
The subdifferential of $F$ at $x^{k+1}$ is given by
\[
\partial F(x^{k+1}) = f'(h(x^{k+1})) \partial h(x^{k+1}) + \nabla g(x^{k+1}) + N_\Omega(x^{k+1}).
\]
Define the vector
\[\begin{aligned} 
v_{k+1} :=  \beta (x^{k+1} - x^{k}) - \nabla g(x^{k}) - f'(h(x^{k})) s_{k+1}  + \nabla g(x^{k+1}) + f'(h(x^{k+1})) s_{k+1}.
 \end{aligned}
\]
By construction, $v_{k+1} \in \partial F(x^{k+1})$. Rearranging terms,
\[\begin{aligned} 
v_{k+1} = \beta (x^{k+1} - x^{k}) + (\nabla g(x^{k+1}) - \nabla g(x^{k}))  + (f'(h(x^{k+1})) - f'(h(x^{k}))) s_{k+1}.
 \end{aligned}
\]
Taking norms and using the Lipschitz continuity of $\nabla g$, $f'$, and $h$, as well as the boundedness of $\{s_{k+1}\}$ (which follows from the local boundedness of $\partial h$), we obtain
\[\begin{aligned} 
\|v_{k+1}\| \le  \beta \|x^{k+1} - x^{k}\| + L \|x^{k+1} - x^{k}\| + L_f L_h \|x^{k+1} - x^{k}\| \|s_{k+1}\|.
  \end{aligned}
\]
Since $\|s_{k+1}\|$ is bounded by some constant $M > 0$, we have
\[
\|v_{k+1}\| \le (\beta + L + L_f L_h M) \|x^{k+1} - x^{k}\|.
\]
Thus, $\mathrm{dist}(0, \partial F(x^{k+1})) \le \|v_{k+1}\| \le C \|x^{k+1} - x^{k}\|$ with $C = \beta + L + L_f L_h M$.
\end{proof}

With the sufficient decrease property and the relative error estimate, we can now establish the full convergence result under the KL assumption.

\begin{theorem}[Global Convergence under KL]
Assume that $F$ is bounded below and satisfies the KL property. Let $\{x^{k}\}$ be generated by Algorithm~\ref{alg.1} with $\beta > L$. Then:

\begin{enumerate}
    \item [(i)] $\sum_{k=0}^\infty \|x^{k+1} - x^{k}\| < \infty$, and consequently $\{x^{k}\}$ converges to a stationary point $x^*$ of \eqref{prob.general}.
    \item [(ii)] Let $\theta \in [0,1)$ be the KL exponent of $F$ at $x^*$. Then the following convergence rates hold:
    \begin{itemize}
        \item If $\theta = 0$, then $\{x^{k}\}$ converges in a finite number of steps.
        \item If $\theta \in (0, \frac{1}{2}]$, then there exist constants $C > 0$ and $q \in (0,1)$ such that $\|x^{k} - x^*\| \le C q^k$ (R-linear convergence) \cite{attouch2013convergence, chen2024enhancing}.
        \item If $\theta \in (\frac{1}{2}, 1)$, then there exists a constant $C > 0$ such that $\|x^{k} - x^*\| \le C k^{-\frac{1-\theta}{2\theta-1}}$ (sublinear convergence) \cite{attouch2013convergence}.
    \end{itemize}
\end{enumerate}
\end{theorem}

\begin{proof}
From Theorem~\ref{thm.decent}, we have
\[
F(x^{k}) - F(x^{k+1}) \ge \frac{\beta-L}{2} \|x^{k+1} - x^{k}\|^2.
\]
Summing from $k=0$ to $N$ and using the lower boundedness of $F$ yields
\begin{equation}\label{eq.notag1} 
\sum_{k=0}^\infty \|x^{k+1} - x^{k}\|^2 < \infty.
\end{equation} 

The relative error estimate (Lemma~\ref{lemma:rel_err_est}) gives
\begin{equation}\label{eq.notag2} 
\mathrm{dist}(0, \partial F(x^{k+1})) \le C \|x^{k+1} - x^{k}\|.
\end{equation} 

With \eqref{eq.notag1}, \eqref{eq.notag2} and the sufficient decrease property Theorem~\ref{thm.decent}, the sequence satisfies the three conditions required for the KL convergence framework \cite{attouch2009convergence}:
\begin{enumerate}
    \item[(i)] Sufficient decrease: $F(x^{k+1}) + a\|x^{k+1} - x^{k}\|^2 \le F(x^{k})$;
    \item[(ii)] Relative error: $\mathrm{dist}(0, \partial F(x^{k+1})) \le b\|x^{k+1} - x^{k}\|$;
    \item[(iii)] Continuity: $F$ is continuous on its domain.
\end{enumerate}

By the KL property and the abstract convergence theorem \cite[Theorem 2.9]{attouch2013convergence}, we obtain that $\sum_{k=0}^\infty \|x^{k+1} - x^{k}\| < \infty$, which implies that $\{x^{k}\}$ is a Cauchy sequence and hence converges to some limit $x^*$. The optimality of $x^*$ follows from taking limits as shown in  Theorem~\ref{thm.convergence}(iv).

The convergence rates follow from the KL exponent $\theta$ via the standard KL analysis \cite[Theorem 2]{attouch2009convergence}. The finite convergence for $\theta=0$ occurs because the KL inequality becomes sharp enough to force the sequence to reach the critical point exactly. For $\theta \in (0, 1/2]$, the resulting recursion yields linear convergence, while for $\theta > 1/2$, the rate becomes sublinear \cite{attouch2009convergence}.
\end{proof}

\section{Numerical experiment}
\label{sec:experiments}

In this section, we present numerical experiments to illustrate the performance of the proposed algorithm and the $\ell_1/\ell_p^q$ model. Our experiments are divided into three major parts:  
(1) \emph{Sparse Signal Recovery}, which provides controlled, synthetic scenarios to thoroughly evaluate the algorithmic behavior,
(2) \emph{Time-domain Signal Reconstruction}, which evaluates the ability of the proposed method to recover structured temporal signals and preserve their key waveform characteristics, and  
(3) \emph{MRI Reconstruction}, which serves as a more realistic imaging application demonstrating the practical value of our approach.

All experiments were performed on a Mac Mini equipped with an M4 chip, using MATLAB R2022a. The initial point for our algorithm is obtained from the solution of the $\ell_1$ minimization problem, computed using the commercial solver Gurobi, ensuring a fair and consistent initialization across all tested methods. 
Performance comparisons are conducted among the classical $\ell_1$ model, the $\ell_1 - \ell_2$ model, and the $\ell_1/\ell_2$ model, together with several variants of our proposed $\ell_1/\ell_p^q$ formulation under different choices of $p$ and $q$, where the $\ell_1/\ell_2$ model can be viewed as a special case of this formulation.

\subsection{Sparse Signal Recovery}

We first evaluate the performance of different models on synthetic sparse signal recovery problems. 
The measurements follow the standard compressed sensing model \( b = A x_{\text{true}} \), 
where $A \in \mathbb{R}^{M \times N}$ is the sensing matrix and $x_{\text{true}}$ is the ground-truth sparse signal. 

In all experiments, the sensing matrix has dimension $64 \times 512$. 
Three types of sensing matrices are considered. 
For the Gaussian case, each row of $A$ is drawn from a multivariate normal distribution $\mathcal{N}(0,\Sigma)$ with covariance matrix $\Sigma = r\mathbf{1}\mathbf{1}^\top + (1-r)I$ and $r=0.8$. 
For the oversampled DCT case, the matrix is generated as
\[
A_{ij}=\frac{1}{\sqrt{M}}\cos\!\left(\frac{2\pi w_i j}{F}\right),
\]
where $w_i \sim \mathcal{U}(0,1)$ and $F=5$ is the oversampling factor. 
For the Bernoulli case, the entries of $A$ are independently drawn from $\{\pm1\}$ with equal probability and scaled by $1/\sqrt{M}$.

The ground-truth signal $x_{\text{true}}\in\mathbb{R}^{N}$ is generated as a $K$-sparse vector whose support is randomly selected with a minimum separation of $2F$ (with $F=5$). 
Two coefficient distributions are considered. 
In the dynamic case, the nonzero entries follow $\mathrm{sign}(\mathcal{N}(0,1))\cdot 10^{u}$ with $u\sim\mathcal{U}(0,3)$, producing a large dynamic range. 
In the non-dynamic case, the nonzero entries are drawn from $\mathcal{N}(0,1)$.

\begin{figure*}[htbp!]
    \centering
    \captionsetup[subfloat]{font=small}
    
    \subfloat[Gaussian, non-dynamic]{
        \includegraphics[width=0.3\linewidth]{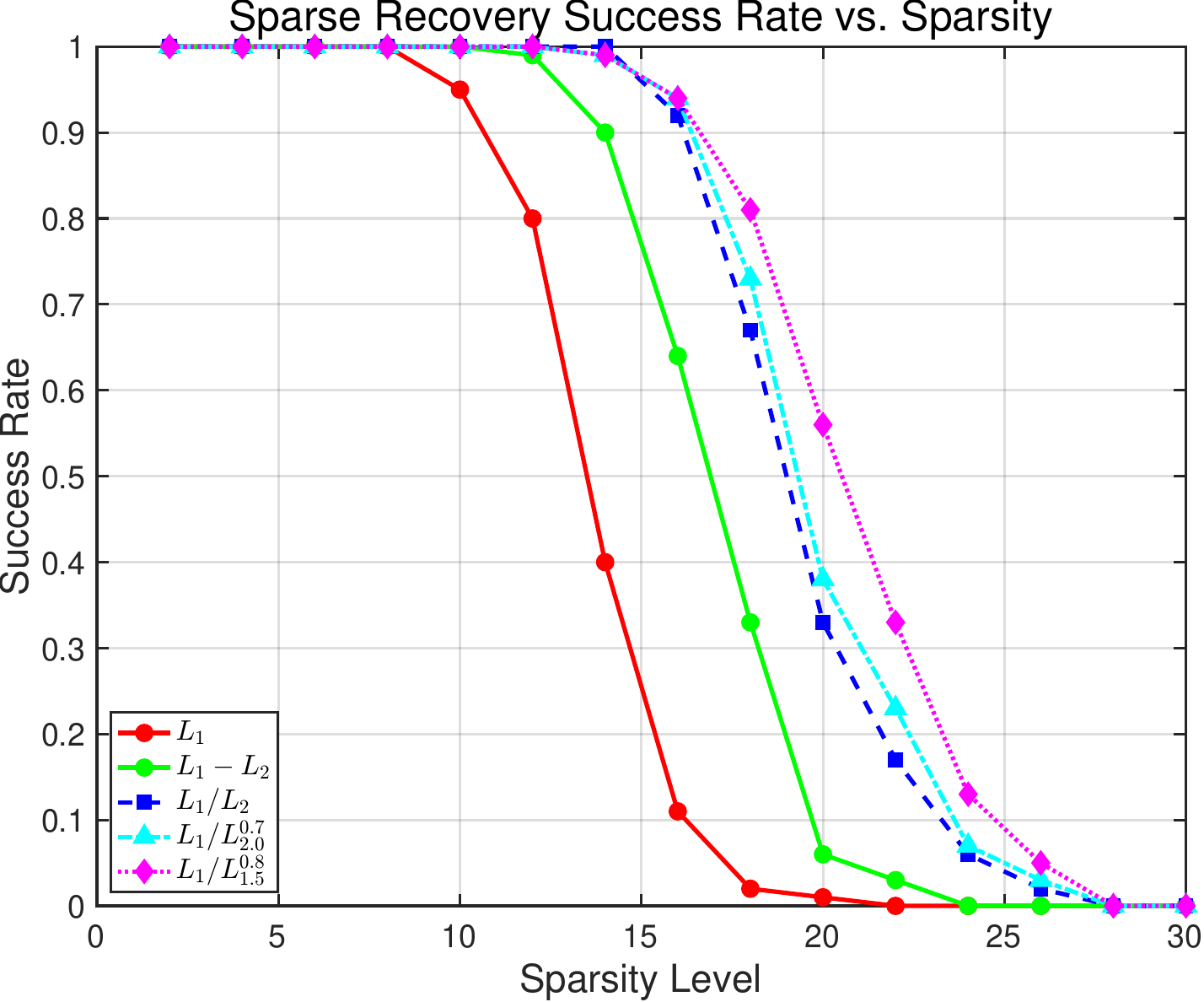}
        \label{fig:Gaussian-non-dynamic}
    }
    \hfill
    \subfloat[Gaussian, dynamic]{
        \includegraphics[width=0.3\linewidth]{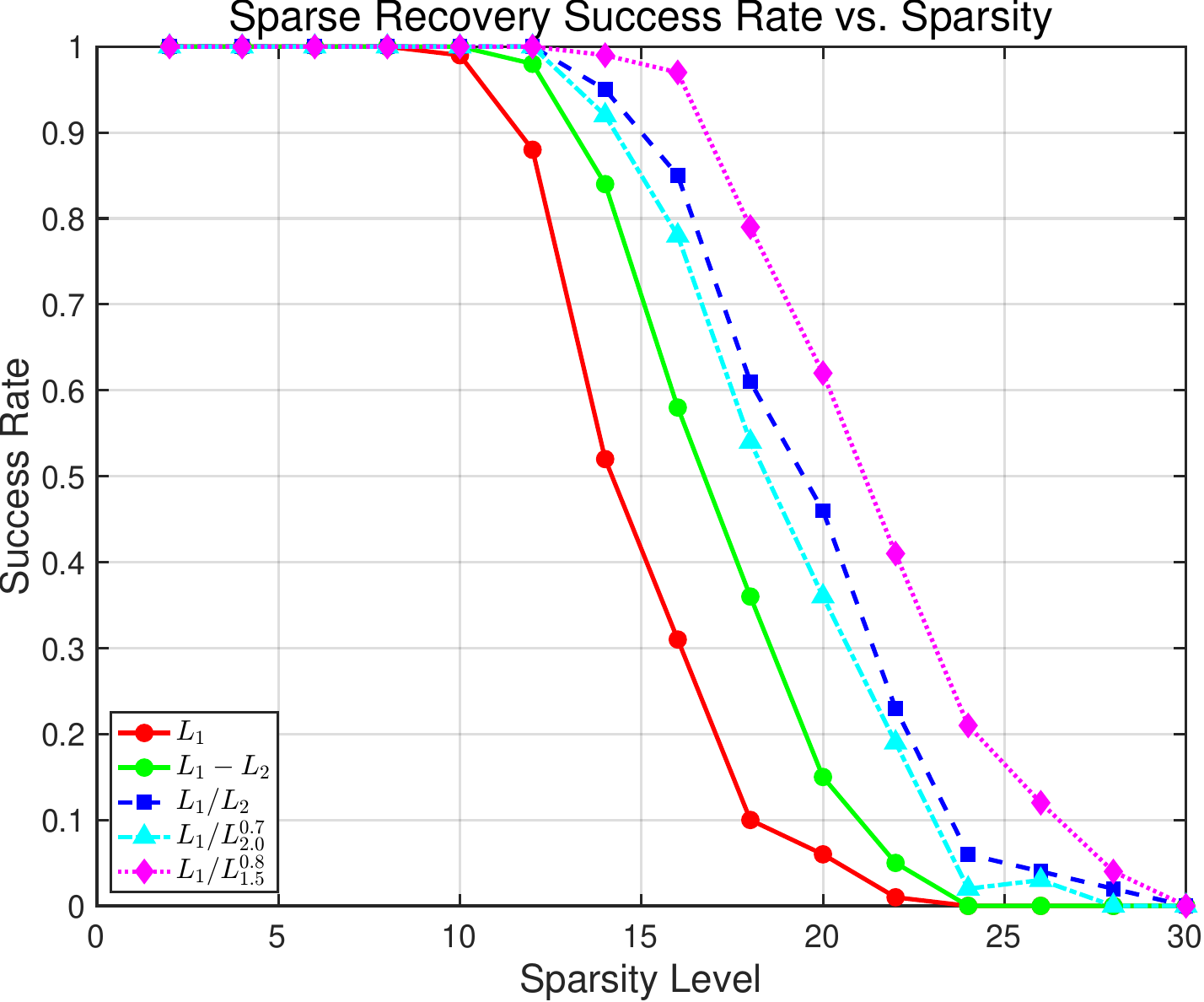}
        \label{fig:Gaussian-dynamic}
    }
    \hfill
    \subfloat[Oversampled DCT, non-dynamic]{
        \includegraphics[width=0.3\linewidth]{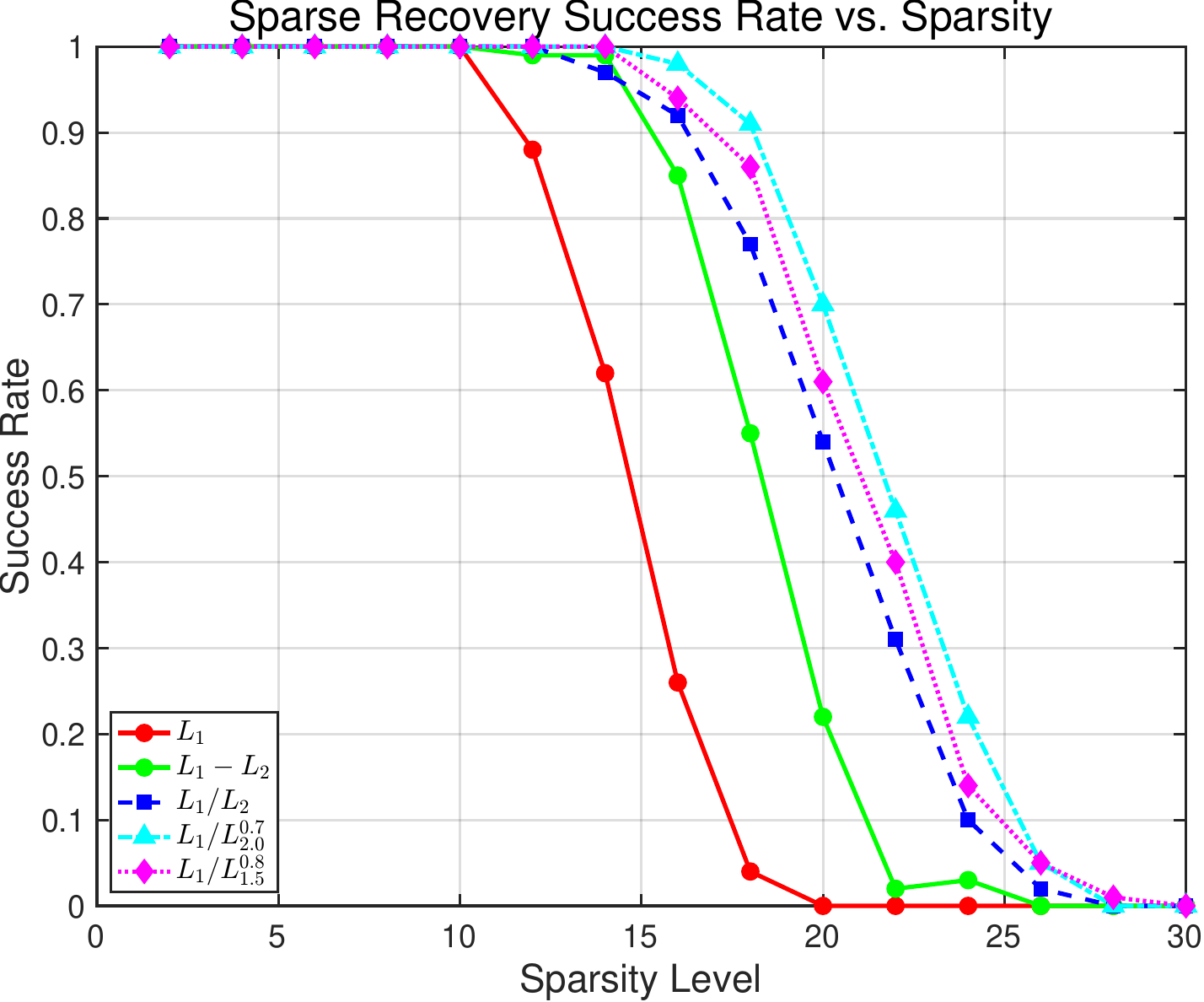}
        \label{fig:Oversample-DCT-non-dynamic}
    }

    \subfloat[Oversampled DCT, dynamic]{
        \includegraphics[width=0.3\linewidth]{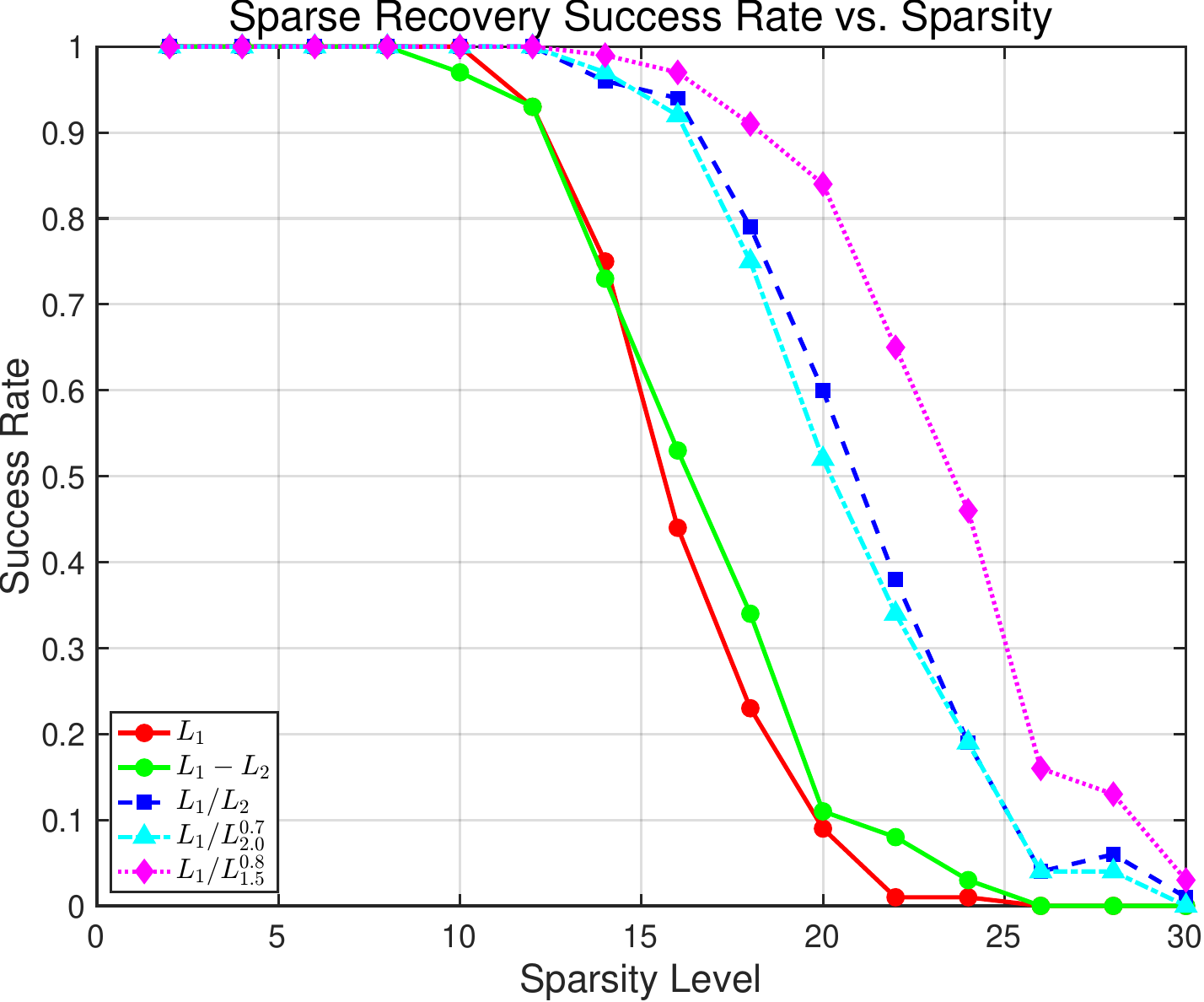}
        \label{fig:Oversample-DCT-dynamic}
    }
    \hfill
    \subfloat[Bernoulli, non-dynamic]{
        \includegraphics[width=0.3\linewidth]{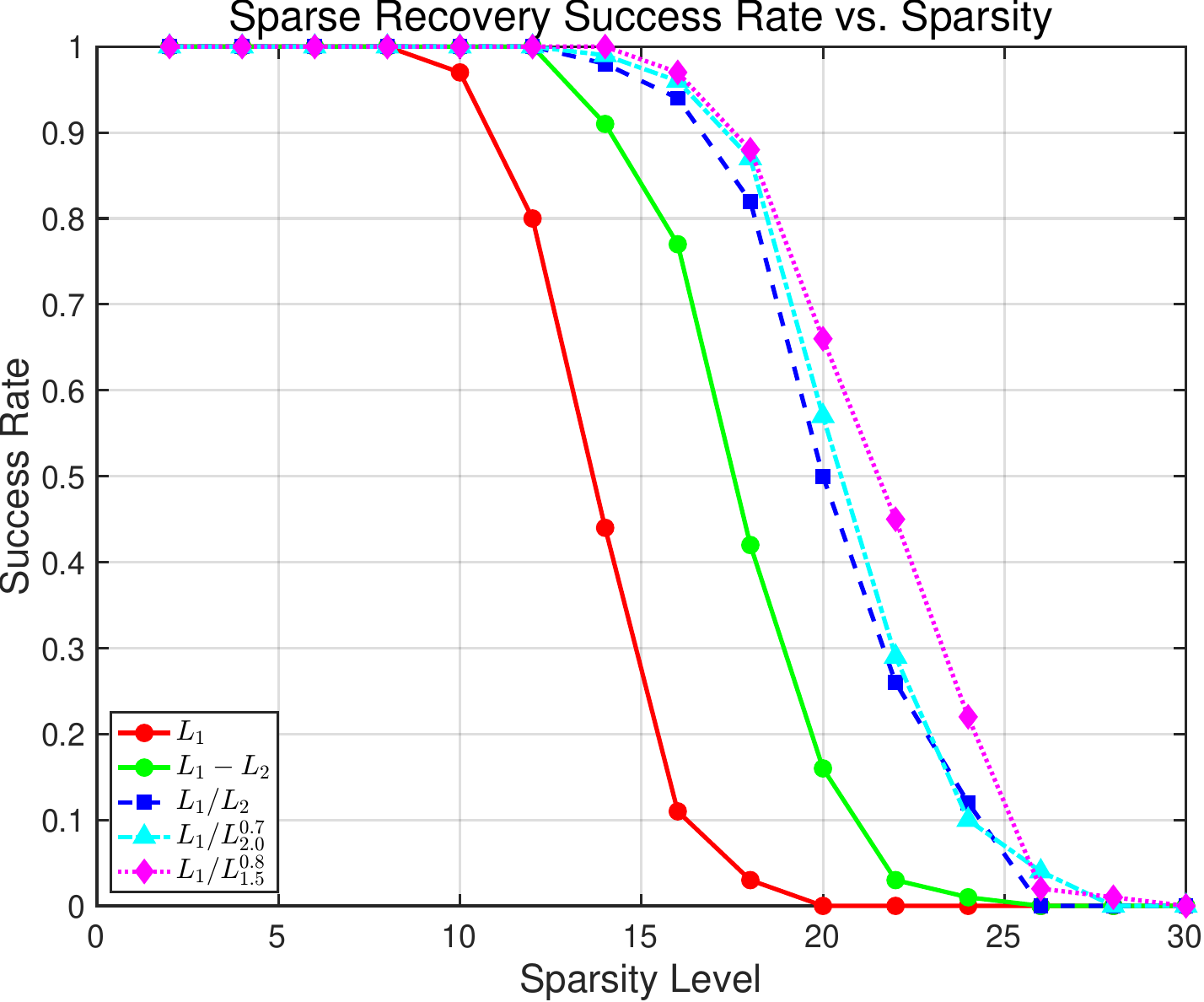}
        \label{fig:Bernoulli-non-dynamic}
    }
    \hfill
    \subfloat[Bernoulli, dynamic]{
        \includegraphics[width=0.3\linewidth]{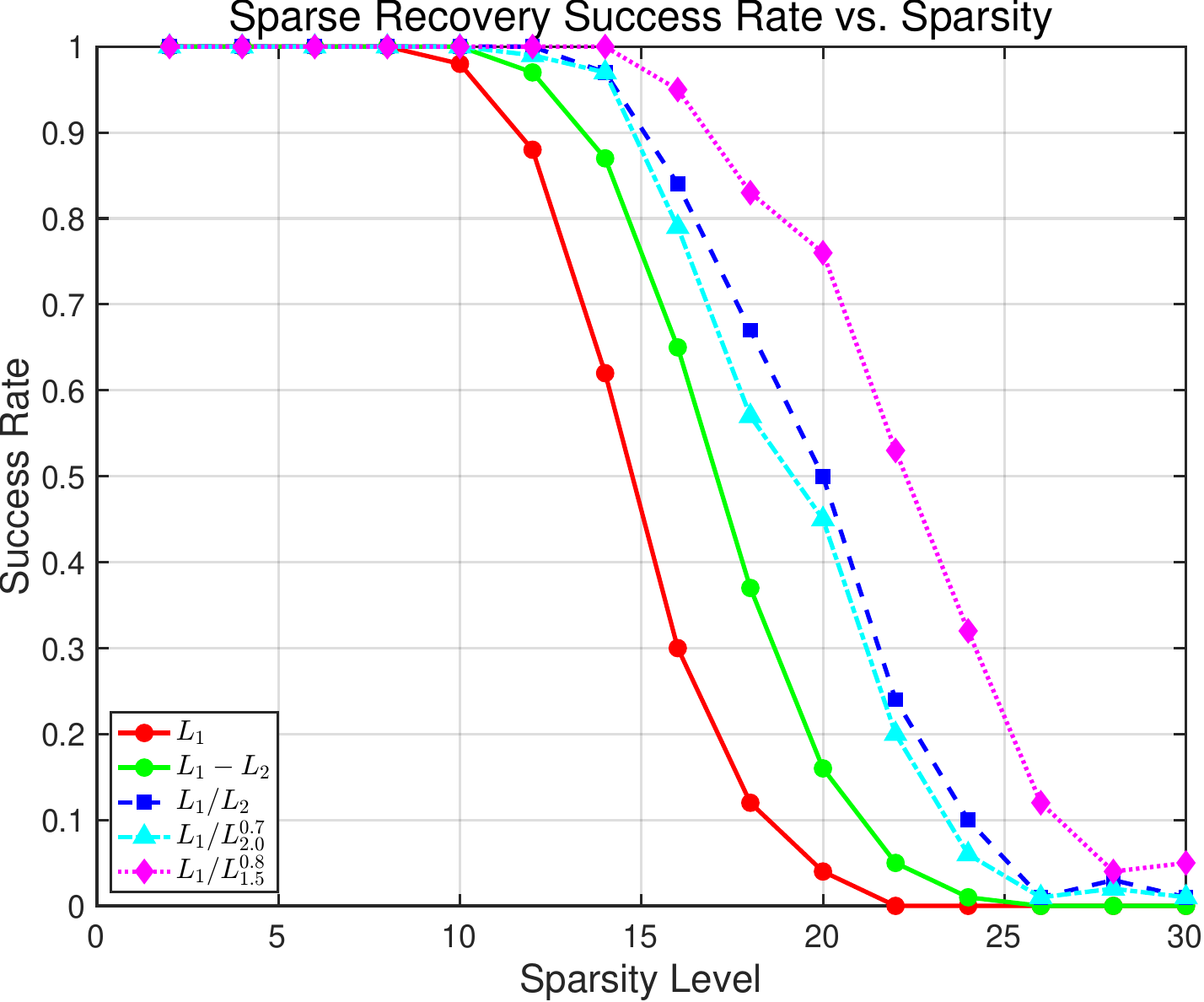}
        \label{fig:Bernoulli-dynamic}
    }
    
    \caption{Success rates of different models for sparse signal recovery under different sensing matrices and sparsity patterns.}
    \label{fig:sparse_recovery_results}
\end{figure*}

We compare five models: the classical $\ell_1$ model, the $\ell_1-\ell_2$ model, the $\ell_1/\ell_2$ model, and two variants of the proposed $\ell_1/\ell_p^q$ formulation with parameter settings $(p,q)=(2.0,0.7)$ and $(p,q)=(1.5,0.8)$. 
The recovery performance is evaluated using the \emph{success rate}. 
A reconstruction $\hat{x}$ is considered successful if the relative error with respect to the ground truth satisfies
\( \frac{\|\hat{x}-x_{\text{true}}\|_2}{\|x_{\text{true}}\|_2} < 10^{-3}. \)
The success rate is defined as the percentage of successful recoveries over 100 independent trials for each configuration.

The success rates for different sparsity levels are shown in Fig.~\ref{fig:sparse_recovery_results}, where the sparsity varies from $2$ to $31$ with step size $2$. 
Overall, the experimental results are consistent with the theoretical prediction, showing the proposed $\ell_1/\ell_p^q$ models achieve the best performance across most experimental configurations.
In particular, the variant with $(p,q)=(1.5,0.8)$ consistently outperforms the other models in all cases except Fig.~\ref{fig:Oversample-DCT-non-dynamic}, where the variant with $(p,q)=(2.0,0.7)$ achieves the highest success rate.
The results also reveal a consistent performance hierarchy among the classical models. 
Except for the dynamic Oversampled DCT case (Fig.~\ref{fig:Oversample-DCT-dynamic}), the experiments generally show that $\ell_1/\ell_2$ performs better than $\ell_1-\ell_2$, which in turn outperforms the classical $\ell_1$ model. 

\subsection{Time-domain Signal Reconstruction}
\label{subsec:time-domain}

We further evaluate the proposed model on a time-domain signal reconstruction task. 
A synthetic signal of length $N=1024$ is constructed as a superposition of several sinusoidal components with different frequencies:
\(
s(t) = \cos\left(\frac{2\pi t}{256}\right) + \sin\left(\frac{2\pi t}{128}\right) + \cos\left(\frac{2\pi t}{64}\right).
\)
This signal consists of a small number of frequency components and is therefore sparse in the frequency domain, leading to a compact representation under the DCT basis. 
Compressive measurements are obtained using a Gaussian random sensing matrix with $M=150$ measurements.

\begin{figure*}[htbp!]
    \centering
    \captionsetup[subfloat]{font=small}
    
    \subfloat[\(\ell_1\) (SNR=26.83)]{
        \includegraphics[width=0.47\linewidth]{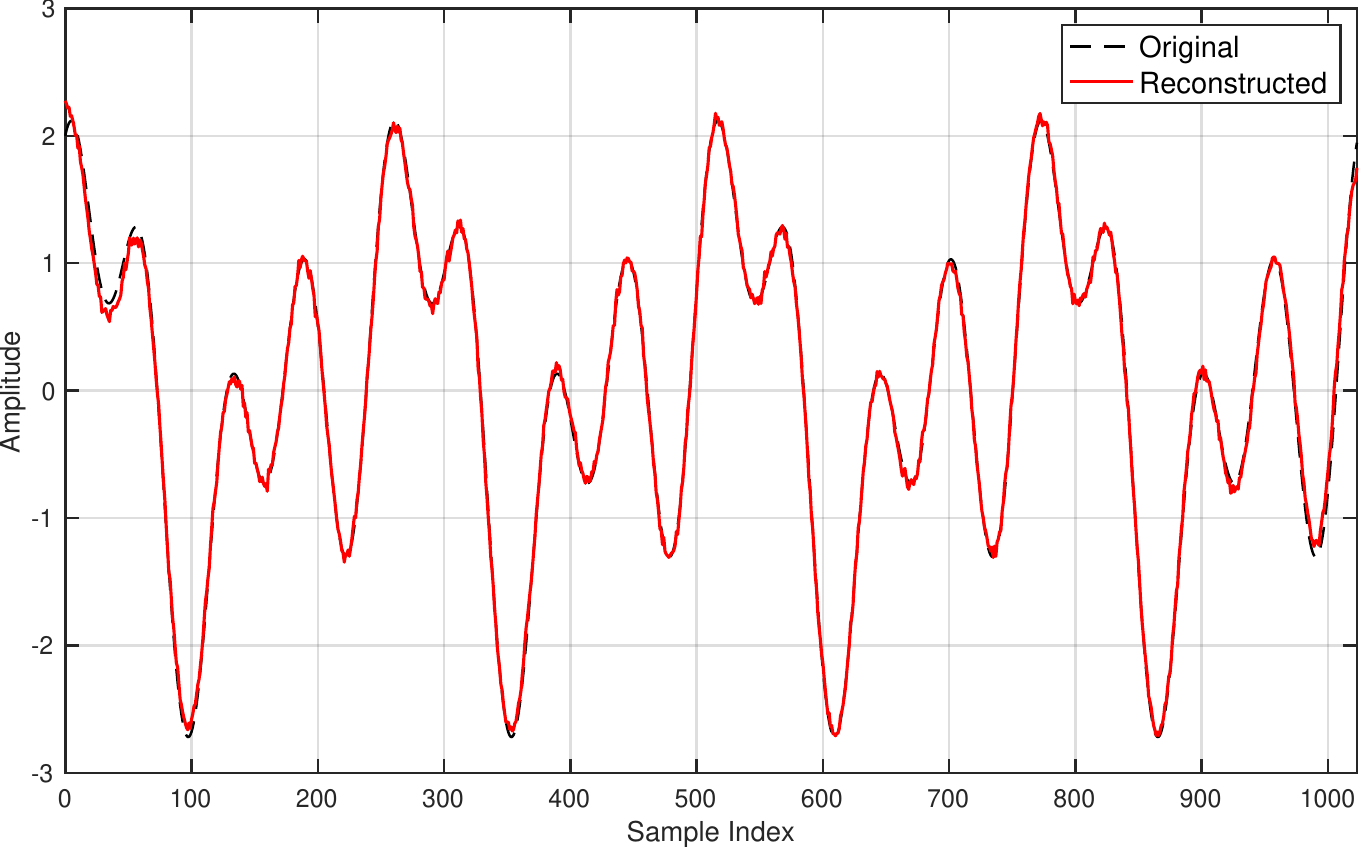}
        \label{fig:time_domain_l1}
    }
    \hfill
    \subfloat[\(\ell_1 - \ell_2\) (SNR=29.14)]{
        \includegraphics[width=0.47\linewidth]{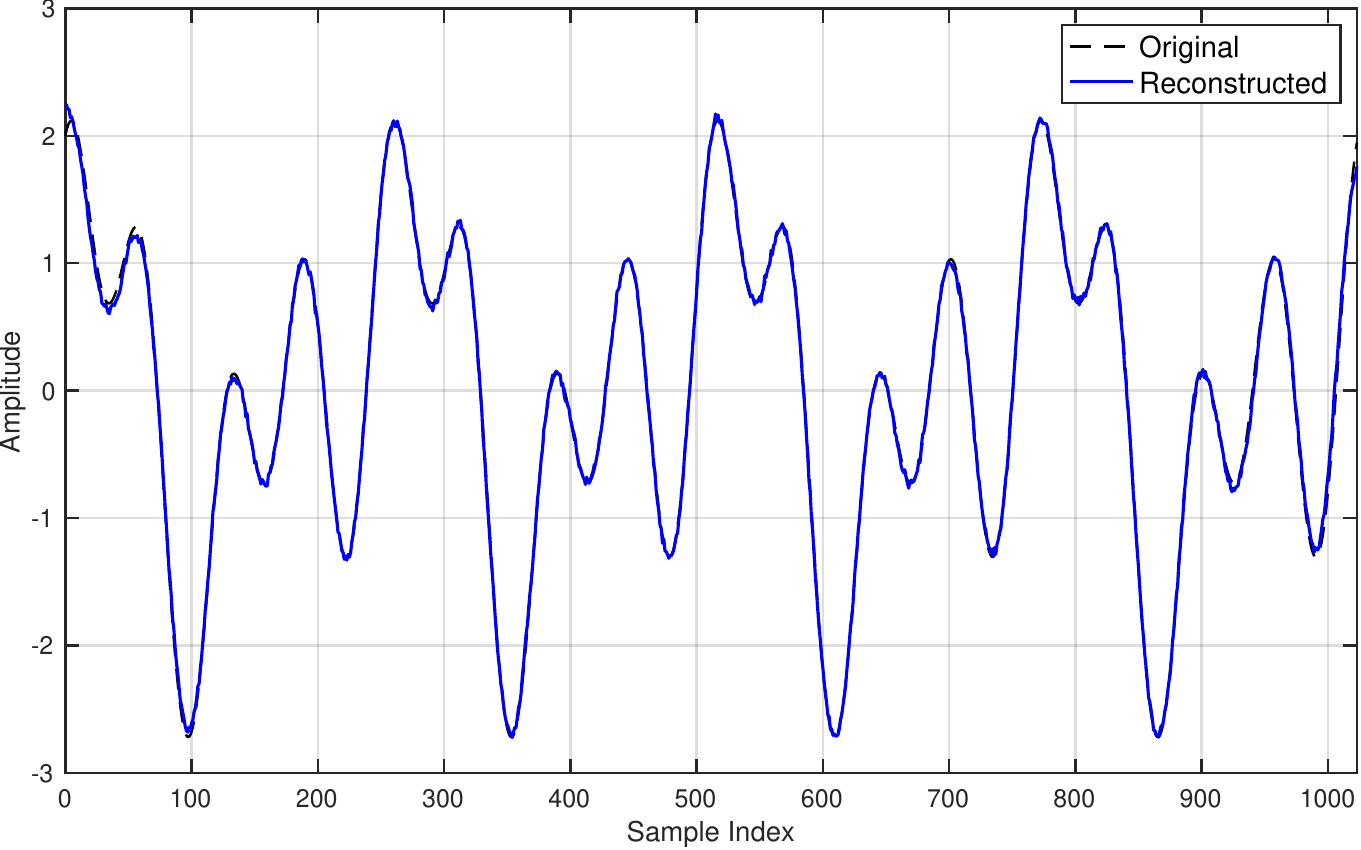}
        \label{fig:time_domain_l1-l2}
    }
    
    \subfloat[\(\ell_1 / \ell_2\) (SNR=30.13)]{
        \includegraphics[width=0.47\linewidth]{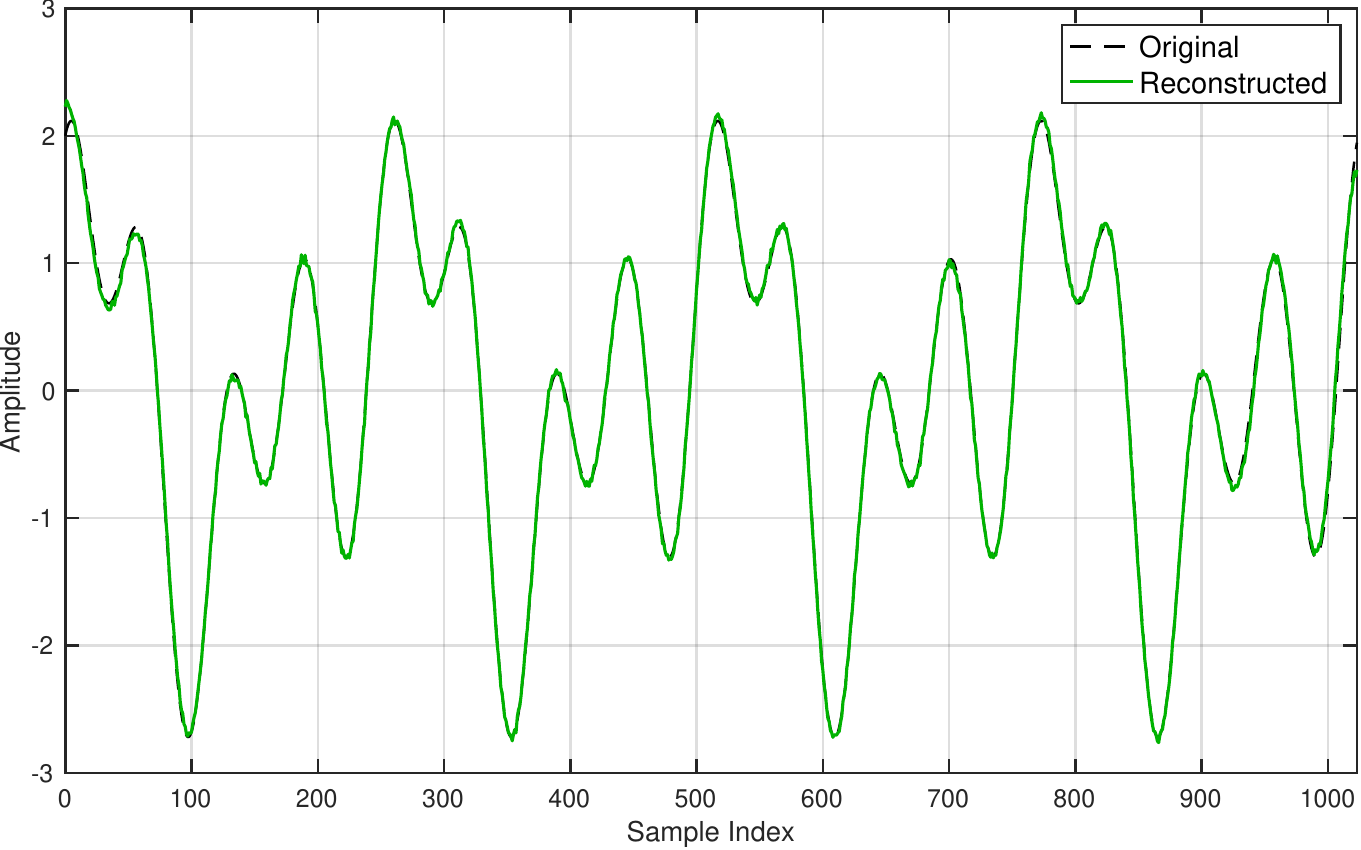}
        \label{fig:time_domain_l1l2}
    }
    \hfill
    \subfloat[\(\ell_1 / \ell_{1.5}^{0.8}\) (SNR=32.44)]{
        \includegraphics[width=0.47\linewidth]{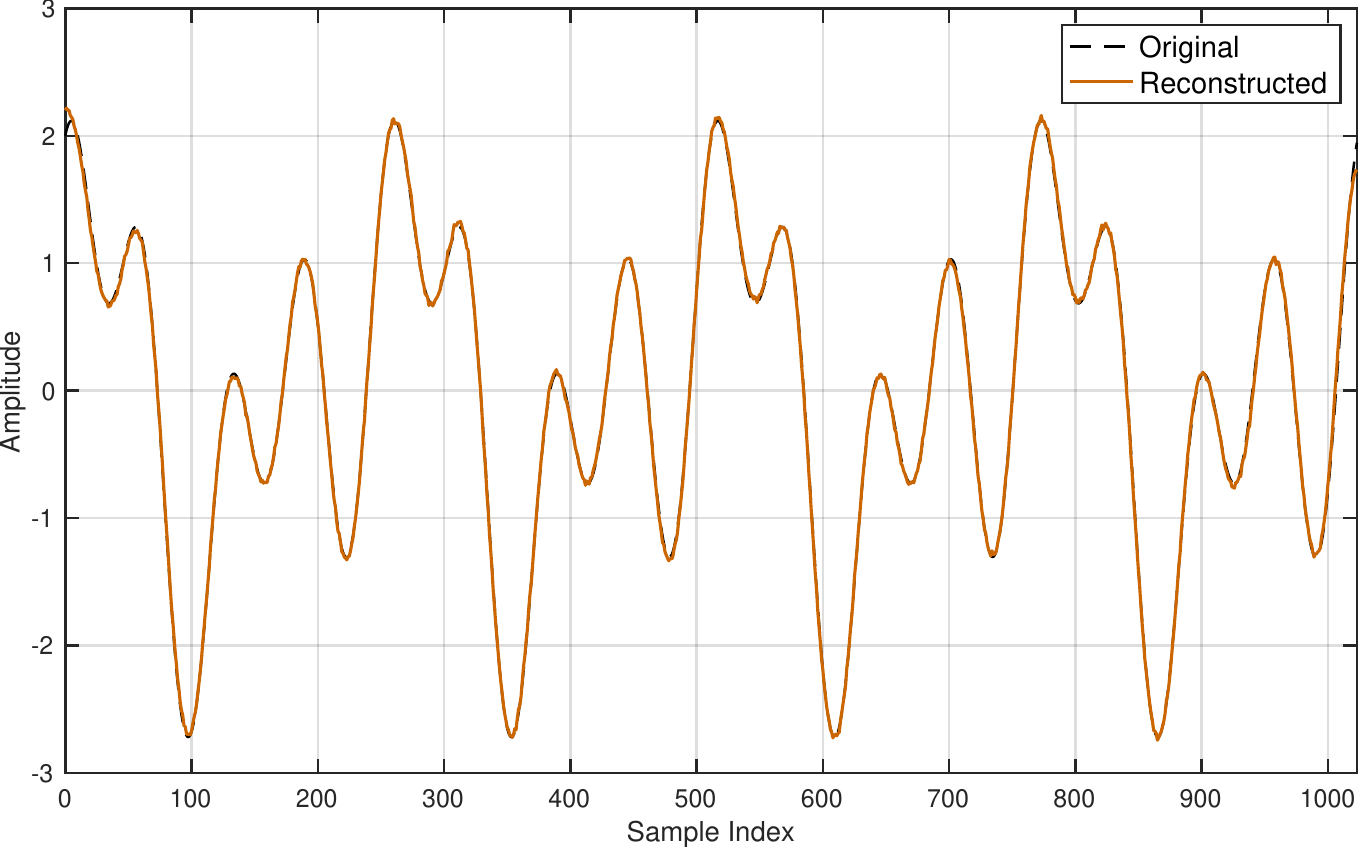}
        \label{fig:time_domain_l1dlpq}
    }
    \caption{Time domain signal reconstruction results.}
    \label{fig:time_domain}
\end{figure*}

To quantitatively evaluate the reconstruction quality, we compute the signal-to-noise ratio (SNR) defined as
\(
\mathrm{SNR} = 20\log_{10}\!\left(\frac{\|s\|_2}{\|s-\hat{s}\|_2}\right),
\)
where $s$ denotes the ground-truth signal and $\hat{s}$ is the reconstructed signal. 
A higher SNR indicates a more accurate reconstruction with smaller recovery error.

We compare four reconstruction models: the classical $\ell_1$ model, the $\ell_1-\ell_2$ model, the $\ell_1/\ell_2$ model, and the proposed $\ell_1/\ell_{1.5}^{0.8}$ model. 
The reconstruction results are shown in Fig.~\ref{fig:time_domain}.
As shown in Fig.~\ref{fig:time_domain}, all methods are able to capture the overall waveform structure. 
However, the proposed $\ell_1/\ell_{1.5}^{0.8}$ model produces the most accurate reconstruction, particularly in preserving the oscillatory patterns and fine-scale variations of the signal, and achieves the highest SNR among all tested models.

\subsection{MRI Reconstruction}
\label{subsec:mri_reconstruction}

As a proof-of-concept validation of the proposed model in real-world imaging applications, we consider MRI reconstruction, where the measurements are acquired in the frequency domain via partial Fourier sampling \cite{lustig2007sparse}. The reconstruction task is to recover an image $u \in \mathbb{R}^{n \times m}$ from under-sampled measurements
\[
\min_{u} \frac{\|\nabla u\|_1}{\|\nabla u\|_p^q} \quad \text{s.t.} \ Au = f, \; u \in [0,1],
\]
where, in implementation, the image $u$ is typically vectorized as an element in $\mathbb{R}^{nm}$. The operator $A = PF$ denotes the partial Fourier sampling operator, where $F$ is the discrete Fourier transform and $P$ is a sampling mask. $\nabla u = (\nabla_x u, \nabla_y u)$ represents the discrete gradient of the image along the horizontal and vertical directions. This formulation follows the common assumption that natural images are approximately sparse in the gradient domain, and can be viewed as an extension of total variation (TV) regularization \cite{rudin1992nonlinear}.

\begin{figure}[htbp!]
    \centering
    \captionsetup[subfloat]{font=small}
    
    \subfloat[Original]{
        \includegraphics[width=0.28\linewidth]{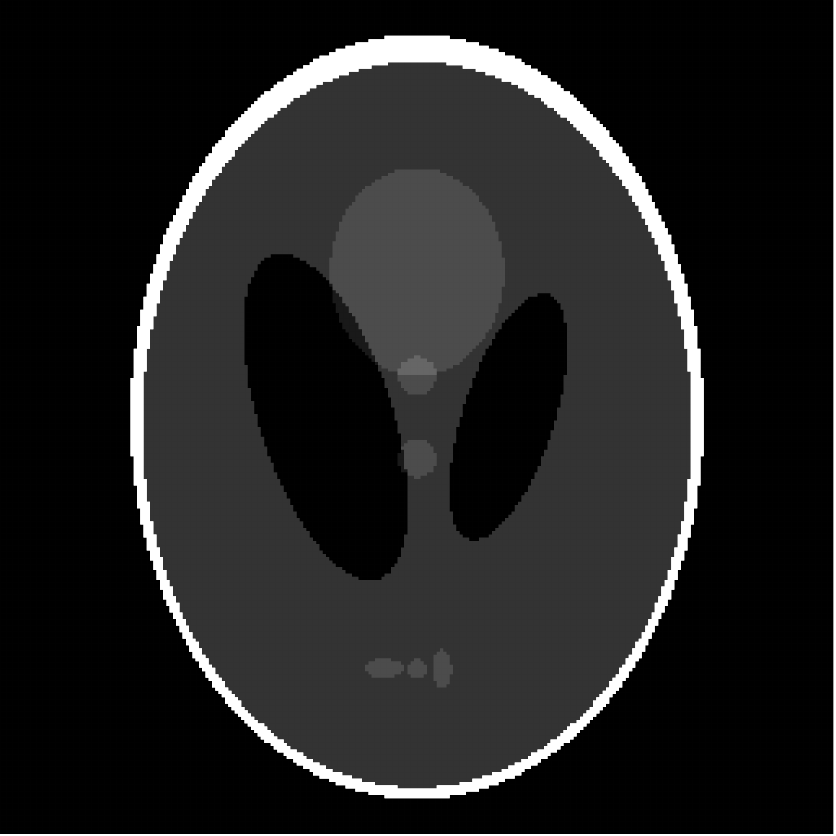}
        \label{fig:mri_original}
    }
    \hfill
    \subfloat[Sampling mask]{
        \includegraphics[width=0.28\linewidth]{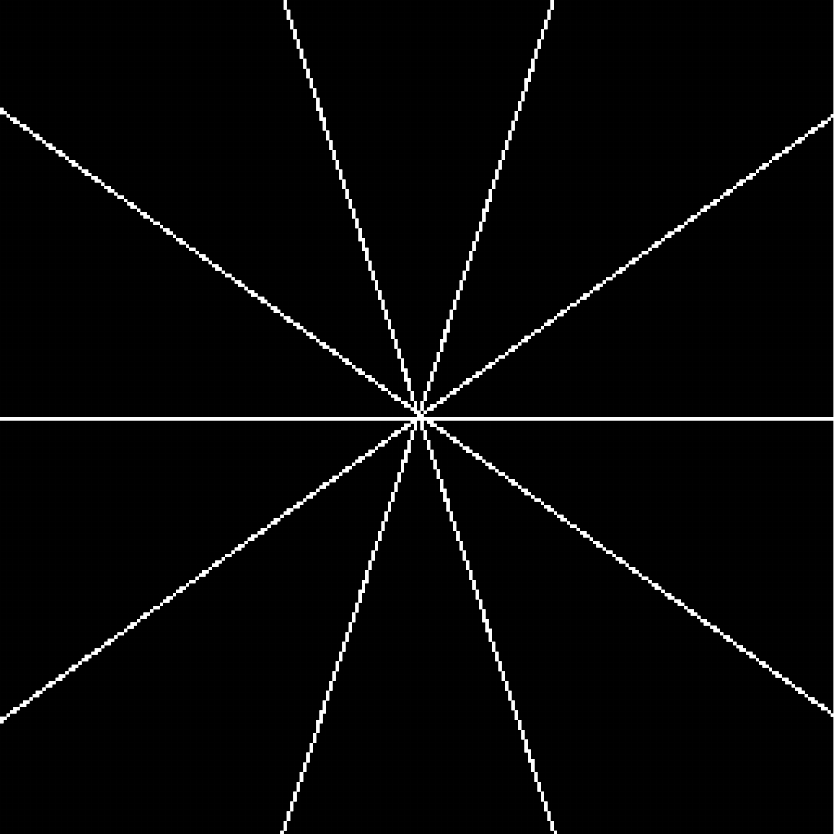}
        \label{fig:mri_mask}
    }
    \hfill
    \subfloat[\(\ell_1\) (RE=53.86\%)]{
        \includegraphics[width=0.28\linewidth]{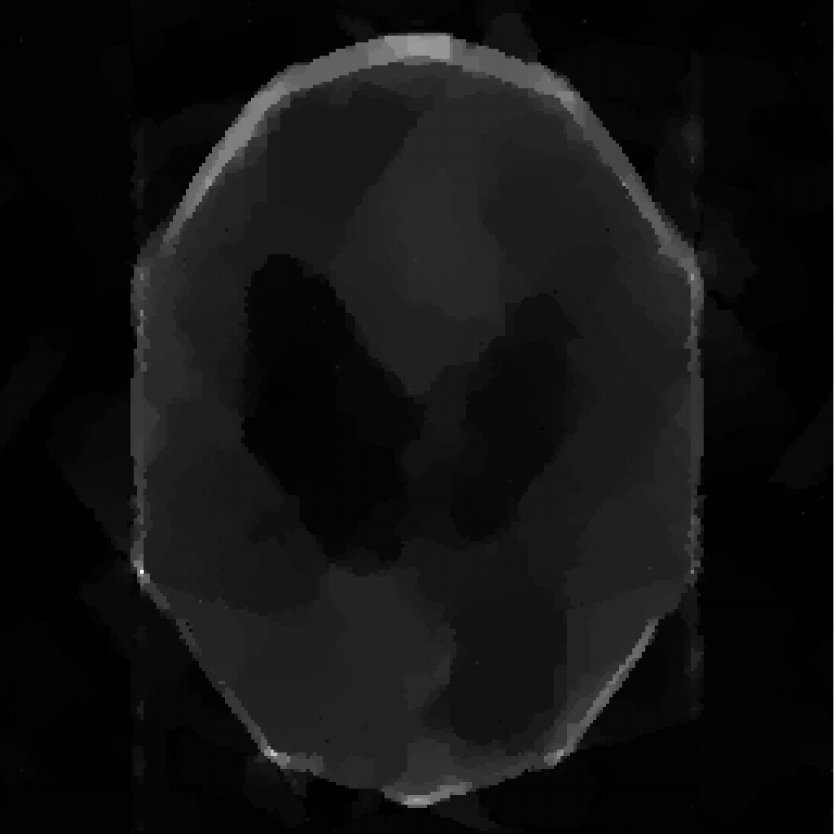}
        \label{fig:mri_l1}
    }

    \subfloat[\(\ell_1 / \ell_2\) (RE=10.30\%)]{
        \includegraphics[width=0.28\linewidth]{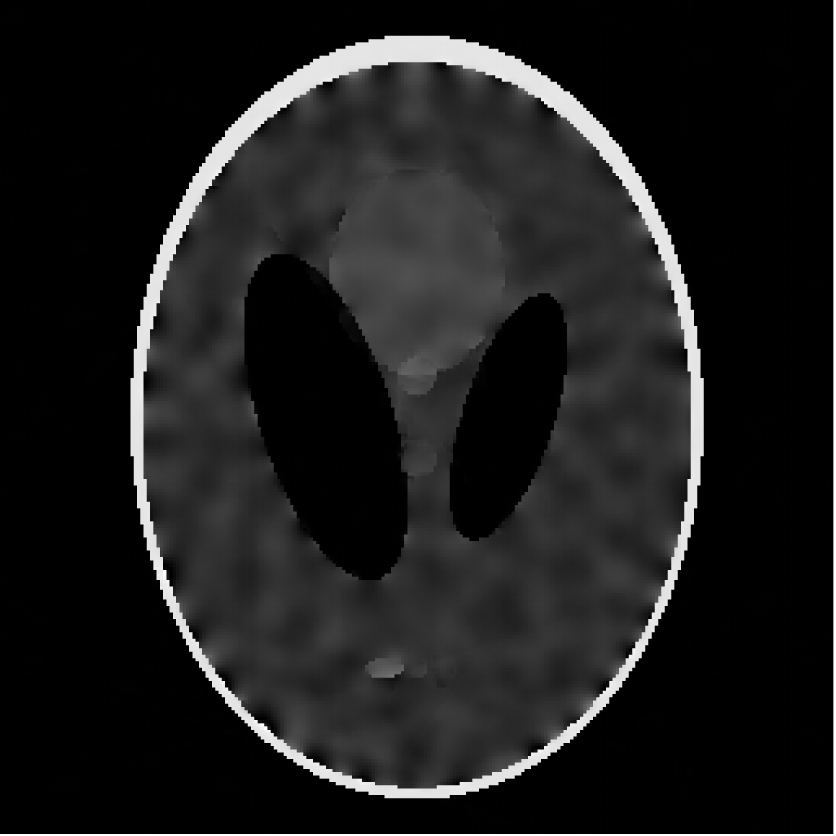}
        \label{fig:mri_l1dl2}
    }
    \hfill
    \subfloat[\(\ell_1 / \ell_{2.0}^{0.7}\) (RE=0.26\%)]{
        \includegraphics[width=0.28\linewidth]{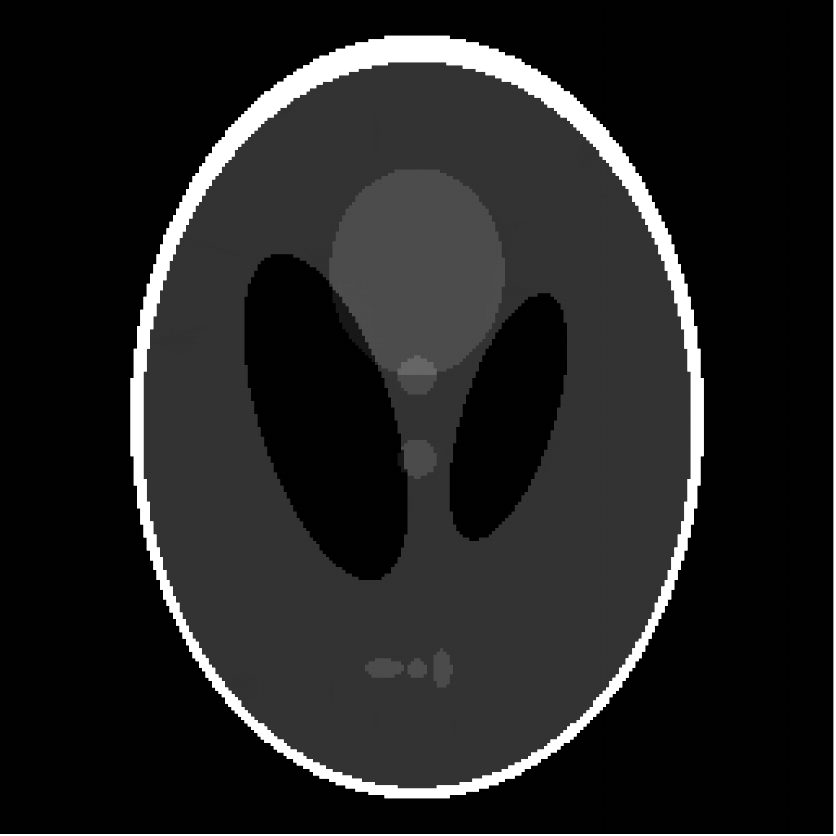}
        \label{fig:mri_l1dlpq1}
    }
    \hfill
    \subfloat[\(\ell_1 / \ell_{2.5}^{0.9}\) (RE=0.15\%)]{
        \includegraphics[width=0.28\linewidth]{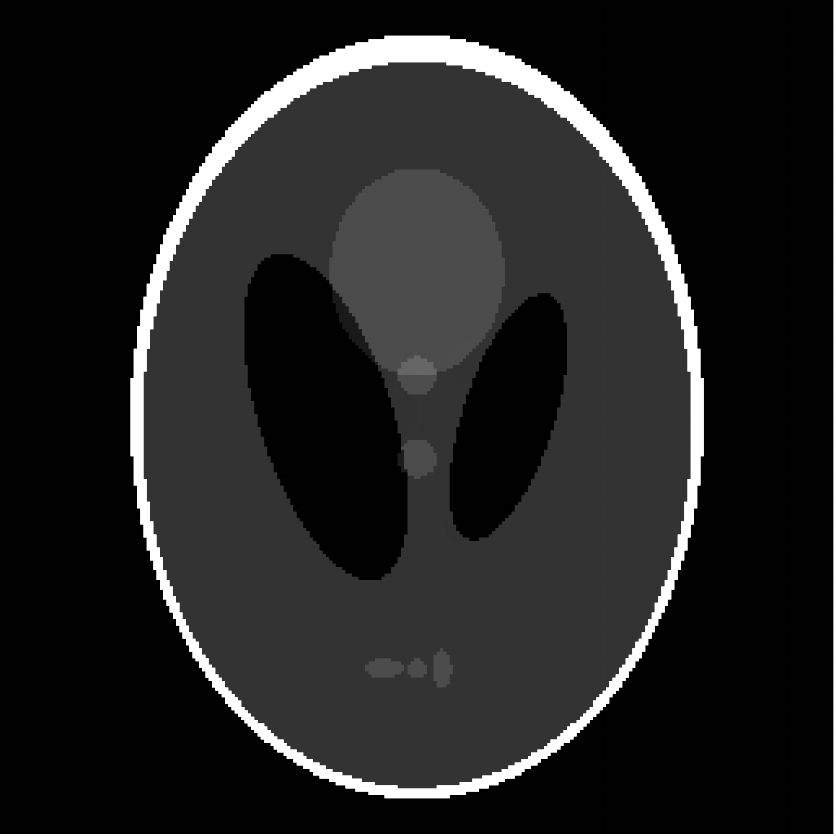}
        \label{fig:mri_l1dlpq2}
    }
    
    \caption{MRI reconstruction results from 5 radial lines in the frequency domain.}
    \label{fig:mri_comparison}
\end{figure}

In our experiments, the $256\times256$ Shepp--Logan phantom is used as the ground truth image, shown in Fig.~\ref{fig:mri_original}. The MRI sampling process is simulated using a radial sampling pattern in the frequency domain, where five radial lines are selected, as illustrated in Fig.~\ref{fig:mri_mask}. Consistent with our constrained formulation, no additional noise is added. The resulting optimization problem is solved using the same framework as described in the previous sections, since the structure remains analogous after introducing the gradient operator.

We compare the classical $\ell_1$ model, the $\ell_1/\ell_2$ model, and two variants of the proposed $\ell_1/\ell_p^q$ formulation with parameter settings $(p,q)=(2.0,0.7)$ and $(p,q)=(2.5,0.9)$. 
The reconstruction results are shown in Fig.~\ref{fig:mri_comparison}.
The classical $\ell_1$ model produces severe artifacts due to the extremely limited measurements, resulting in a high relative error of $53.86\%$ (Fig.~\ref{fig:mri_l1}). 
The $\ell_1/\ell_2$ model improves the reconstruction quality and reduces the error to $10.30\%$, although visible artifacts remain (Fig.~\ref{fig:mri_l1dl2}). 
In contrast, the proposed $\ell_1/\ell_p^q$ models achieve significantly better reconstructions. The configuration $(p,q)=(2.0,0.7)$ reduces the relative error to $0.26\%$ and produces a nearly artifact-free image (Fig.~\ref{fig:mri_l1dlpq1}). 
The configuration $(p,q)=(2.5,0.9)$ further improves the reconstruction quality, achieving the lowest error of $0.15\%$ (Fig.~\ref{fig:mri_l1dlpq2}). 
These results demonstrate the effectiveness of the proposed model for recovering structured images from highly undersampled measurements.

\section{Conclusion}
\label{sec:conclusion}

In this paper, we proposed a sparse signal recovery model based on the $\ell_1/\ell_p^q$ regularization, and established its connection with the $\ell_1 - \alpha \ell_p$ formulation through the equivalence of first-order stationary points. 
We further derived a sufficient recovery condition under the Restricted Isometry Property (RIP), showing that smaller values of $q$ lead to less restrictive bounds and improved error estimates.
To solve the resulting problem, we developed an efficient majorization--minimization (MM) algorithm and established its convergence via the Kurdyka--Łojasiewicz (KL) property. 
Numerical experiments demonstrate that the proposed method achieves superior performance across a range of representative sparse recovery tasks.
Although smaller $q$ improves theoretical recovery guarantees, it may hinder algorithmic performance. 
Future work will aim to improve both the modeling capability and the robustness of the optimization algorithm.

\bibliographystyle{siamplain}
\bibliography{reference.bib}

@article{attouch2009convergence,
  title={On the convergence of the proximal algorithm for nonsmooth functions involving analytic features},
  author={Attouch, Hedy and Bolte, J{\'e}r{\^o}me},
  journal={Mathematical Programming},
  volume={116},
  number={1},
  pages={5--16},
  year={2009},
  publisher={Springer}
}

@article{attouch2013convergence,
  title={Convergence of descent methods for semi-algebraic and tame problems: proximal algorithms, forward--backward splitting, and regularized Gauss--Seidel methods},
  author={Attouch, Hedy and Bolte, J{\'e}r{\^o}me and Svaiter, Benar Fux},
  journal={Mathematical programming},
  volume={137},
  number={1},
  pages={91--129},
  year={2013},
  publisher={Springer}
}

@article{hoyer2004non,
  title={Non-negative matrix factorization with sparseness constraints},
  author={Hoyer, Patrik O},
  journal={Journal of machine learning research},
  volume={5},
  number={Nov},
  pages={1457--1469},
  year={2004}
}

@article{chen2024enhancing,
  title={Enhancing Convergence of Decentralized Gradient Tracking under the KL Property},
  author={Chen, Xiaokai and Cao, Tianyu and Scutari, Gesualdo},
  journal={arXiv preprint arXiv:2412.09556},
  year={2024}
}

@article{zhan2025p,
  title={$L_{p}/L_{q}$ Minimization Method and Stable Recovery of Approximately Block k-Sparse Signals},
  author={Zhan, Kunsheng and Wan, Anhua},
  journal={IEEE Transactions on Signal Processing},
  volume={73},
  pages={2399--2416},
  year={2025},
  publisher={IEEE}
}

@article{fan2001variable,
  title     = {Variable selection via nonconcave penalized likelihood and its oracle properties},
  author    = {Fan, Jianqing and Li, Runze},
  journal   = {Journal of the American statistical Association},
  volume    = {96},
  number    = {456},
  pages     = {1348--1360},
  year      = {2001},
  publisher = {Taylor \& Francis}
}

@misc{gurobi,
  author = {{Gurobi Optimization, LLC}},
  title = {{Gurobi Optimizer Reference Manual}},
  year = 2026,
  url = "https://www.gurobi.com"
}

@article{rudelson2008sparse,
  title={On sparse reconstruction from Fourier and Gaussian measurements},
  author={Rudelson, Mark and Vershynin, Roman},
  journal={Communications on Pure and Applied Mathematics: A Journal Issued by the Courant Institute of Mathematical Sciences},
  volume={61},
  number={8},
  pages={1025--1045},
  year={2008},
  publisher={Wiley Online Library}
}

@article{baraniuk2007compressive,
  title={Compressive sensing [lecture notes]},
  author={Baraniuk, Richard G},
  journal={IEEE signal processing magazine},
  volume={24},
  number={4},
  pages={118--121},
  year={2007},
  publisher={IEEE}
}

@article{beck2009fast,
  title={A fast iterative shrinkage-thresholding algorithm for linear inverse problems},
  author={Beck, Amir and Teboulle, Marc},
  journal={SIAM journal on imaging sciences},
  volume={2},
  number={1},
  pages={183--202},
  year={2009},
  publisher={SIAM}
}

@article{tao2023study,
  title={Study on $L_1$ over $L_2$ minimization for nonnegative signal recovery},
  author={Tao, Min and Zhang, Xiao-Ping},
  journal={Journal of Scientific Computing},
  volume={95},
  number={3},
  pages={94},
  year={2023},
  publisher={Springer}
}

@article{jia2025sparse,
  title={Sparse recovery: The square of $\ell_1/\ell_2$ norms},
  author={Jia, Jianqing and Prater-Bennette, Ashley and Shen, Lixin and Tripp, Erin E},
  journal={Journal of Scientific Computing},
  volume={102},
  number={1},
  pages={24},
  year={2025},
  publisher={Springer}
}

@article{RN471,
  title={A scale-invariant approach for sparse signal recovery},
  author={Rahimi, Yaghoub and Wang, Chao and Dong, Hongbo and Lou, Yifei},
  journal={SIAM Journal on Scientific Computing},
  volume={41},
  number={6},
  pages={A3649--A3672},
  year={2019},
  publisher={SIAM}
}

@article{RN547,
  author = {Xie, Yujia and Su, Xinhua and Ge, Huanmin},
  title = {RIP Analysis for $\ell_{1}/\ell_{p}   (p > 1)$ Minimization Method},
  journal = {IEEE Signal Processing Letters},
  volume={30},
  pages={997--1001},
  year={2023},
  publisher={IEEE}
}

@article{RN475,
   author = {Wang, Chao and Tao, Min and Chuah, Chen-Nee and Nagy, James and Lou, Yifei},
   title = {Minimizing $L_1$ over $L_2$ norms on the gradient},
   journal = {Inverse Problems},
   volume = {38},
   number = {6},
   pages = {065011},
   ISSN = {0266-5611},
   year = {2022},
   type = {Journal Article}
}

@article{RN493,
  title={Minimization of $L_1$ Over $L_2$ for Sparse Signal Recovery with Convergence Guarantee},
  author={Tao, Min},
  journal={SIAM Journal on Scientific Computing},
  volume={44},
  number={2},
  pages={A770--A797},
  year={2022},
  publisher={SIAM}
}

@article{RN540,
   author = {Zeng, Liaoyuan and Yu, Peiran and Pong, Ting Kei},
   title = {Analysis and algorithms for some compressed sensing models based on L1/L2 minimization},
   journal = {SIAM Journal on Optimization},
   volume = {31},
   number = {2},
   pages = {1576-1603},
   ISSN = {1052-6234},
   year = {2021},
   type = {Journal Article}
}

@article{RN465,
  title={Accelerated schemes for the $ L_1/L_2 $ minimization},
  author={Wang, Chao and Yan, Ming and Rahimi, Yaghoub and Lou, Yifei},
  journal={IEEE Transactions on Signal Processing},
  volume={68},
  pages={2660--2669},
  year={2020},
  publisher={IEEE}
}

@article{Bolte2014,
  title={Proximal alternating linearized minimization for nonconvex and nonsmooth problems},
  author={Bolte, J{\'e}r{\^o}me and Sabach, Shoham and Teboulle, Marc},
  journal={Mathematical Programming},
  volume={146},
  number={1-2},
  pages={459--494},
  year={2014},
  publisher={Springer}
}

@article{Bolte2010,
  title={Proximal alternating minimization and projection methods for nonconvex problems: An approach based on the Kurdyka-{\L}ojasiewicz inequality},
  author={Attouch, H{\'e}dy and Bolte, J{\'e}r{\^o}me and Redont, Patrick and Soubeyran, Antoine},
  journal={Mathematics of operations research},
  volume={35},
  number={2},
  pages={438--457},
  year={2010},
  publisher={INFORMS}
}

@article{natarajan1995sparse,
  title={Sparse approximate solutions to linear systems},
  author={Natarajan, Balas Kausik},
  journal={SIAM journal on computing},
  volume={24},
  number={2},
  pages={227--234},
  year={1995},
  publisher={SIAM}
}

@article{chen2001atomic,
  title={Atomic decomposition by basis pursuit},
  author={Chen, Scott Shaobing and Donoho, David L and Saunders, Michael A},
  journal={SIAM review},
  volume={43},
  number={1},
  pages={129--159},
  year={2001},
  publisher={SIAM}
}

@article{candes2005decoding,
  title={Decoding by linear programming},
  author={Candes, Emmanuel J and Tao, Terence},
  journal={IEEE transactions on information theory},
  volume={51},
  number={12},
  pages={4203--4215},
  year={2005},
  publisher={IEEE}
}

@article{chartrand2007exact,
  title={Exact reconstruction of sparse signals via nonconvex minimization},
  author={Chartrand, Rick},
  journal={IEEE Signal Processing Letters},
  volume={14},
  number={10},
  pages={707--710},
  year={2007},
  publisher={IEEE}
}

@article{boyd2011distributed,
  title={Distributed optimization and statistical learning via the alternating direction method of multipliers},
  author={Boyd, Stephen and Parikh, Neal and Chu, Eric and Peleato, Borja and Eckstein, Jonathan and others},
  journal={Foundations and Trends{\textregistered} in Machine learning},
  volume={3},
  number={1},
  pages={1--122},
  year={2011},
  publisher={Now Publishers, Inc.}
}

@article{huo_boldsymboll_1-beta_2023,
  title={$L_1 - \beta L_q$ Minimization for signal and image recovery},
  author={Huo, Limei and Chen, Wengu and Ge, Huanmin and Ng, Michael K},
  journal={SIAM Journal on Imaging Sciences},
  volume={16},
  number={4},
  pages={1886--1928},
  year={2023},
  publisher={SIAM}
}

@article{lustig2007sparse,
  title={Sparse MRI: The application of compressed sensing for rapid MR imaging},
  author={Lustig, Michael and Donoho, David and Pauly, John M},
  journal={Magnetic Resonance in Medicine: An Official Journal of the International Society for Magnetic Resonance in Medicine},
  volume={58},
  number={6},
  pages={1182--1195},
  year={2007},
  publisher={Wiley Online Library}
}

@article{rudin1992nonlinear,
  title={Nonlinear total variation based noise removal algorithms},
  author={Rudin, Leonid I and Osher, Stanley and Fatemi, Emad},
  journal={Physica D: nonlinear phenomena},
  volume={60},
  number={1-4},
  pages={259--268},
  year={1992},
  publisher={Elsevier}
}

@article{dinkelbach1967nonlinear,
  title={On nonlinear fractional programming},
  author={Dinkelbach, Werner},
  journal={Management science},
  volume={13},
  number={7},
  pages={492--498},
  year={1967},
  publisher={INFORMS}
}

@incollection{frenk2005fractional,
  title={Fractional programming},
  author={Frenk, Johannes BG and Schaible, Siegfried},
  booktitle={Handbook of generalized convexity and generalized monotonicity},
  pages={335--386},
  year={2005},
  publisher={Springer}
}

\end{document}